\documentclass[10pt,journal,compsoc]{IEEEtran}

\ifCLASSOPTIONcompsoc
  \usepackage[nocompress]{cite}
\else
  \usepackage{cite}
\fi

\usepackage{caption}
\usepackage{subcaption}
\usepackage{url}
\usepackage{amsmath,amssymb,amsfonts}
\usepackage{algorithmic}
\usepackage{graphicx}
\usepackage{textcomp}
\usepackage{xcolor}
\usepackage[vlined, ruled, boxed,linesnumbered]{algorithm2e}
\usepackage{amsthm}
\usepackage{multirow}
\usepackage{balance}
\SetKw{Break}{break}
\SetKwRepeat{Do}{do}{while}%

\ifCLASSINFOpdf
\else
\fi

\begin{document}
\title{CAMIG: Concurrency-Aware Live Migration\\ Management of Multiple Virtual Machines in SDN-enabled Clouds}

\author{TianZhang~He,~\IEEEmembership{}%
		Adel N. Toosi,~\IEEEmembership{Member,~IEEE,}
        Rajkumar~Buyya,~\IEEEmembership{Fellow,~IEEE}%
\IEEEcompsocitemizethanks{\IEEEcompsocthanksitem T.Z. He and R. Buyya are with the CLOUDS Lab, School of Computing and Information Systems, University of Melbourne, Australia.
(E-mail: tianzhangh@student.unimelb.edu.au; rbuyya@unimelb.edu.au)
\IEEEcompsocthanksitem A. N. Toosi is with the Department of Software Systems and Cybersecurity, Faculty of Information Technology, Monash University, Australia. (E-mail: adel.n.toosi@monash.edu)}%
}

\IEEEtitleabstractindextext{%
\begin{abstract}
By integrating Software-Defined Networking and cloud computing, virtualized networking and computing resources can be dynamically reallocated through live migration of Virtual Machines (VMs). Dynamic resource management such as load balancing and energy-saving policies can request multiple migrations when the algorithms are triggered periodically.  
There exist notable research efforts in dynamic resource management that alleviate single migration overheads, such as single migration time and co-location interference while selecting the potential VMs and migration destinations. 
However, by neglecting the resource dependency among potential migration requests, the existing solutions of dynamic resource management can result in the Quality of Service (QoS) degradation and Service Level Agreement (SLA) violations during the migration schedule. Therefore, it is essential to integrate both single and multiple migration overheads into VM reallocation planning. In this paper, we propose a concurrency-aware multiple migration selector that operates based on the maximal cliques and independent sets of the resource dependency graph of multiple migration requests. Our proposed method can be integrated with existing dynamic resource management policies. The experimental results demonstrate that our solution efficiently minimizes migration interference and shortens the convergence time of reallocation by maximizing the multiple migration performance while achieving the objective of dynamic resource management.
\end{abstract}

}

\maketitle

\IEEEdisplaynontitleabstractindextext

\IEEEpeerreviewmaketitle

\IEEEraisesectionheading{\section{Introduction}\label{sec:introduction}}
With the rapid adoption of cloud computing for hosting applications and always-on services, it is critical to provide Quality of Service (QoS) guarantees through the Service Level Agreements (SLAs) between cloud providers and users. In this direction, many research works have investigated various aspects of dynamic resource management, such as delay-aware Virtual Network Function (VNF) placement~\cite{cziva2018dynamic}, load balancing~\cite{wood2009sandpiper, verma2008pmapper, mann2012}, energy-saving~\cite{beloglazov2012optimal}, flow consolidation, scheduled maintenance, as well as emergency migration, in terms of accessibility, quality, efficiency, and robustness of cloud services. Virtual Machine (VM) is one of the major virtualization technologies to host computing and networking resources in cloud data centers. As a dynamic resource management tool, the live VM migration is used to realize the objectives in resource management by relocating VMs between physical hosts without disrupting the accessibility of cloud services \cite{clark2005live}.

Cloud infrastructure and service providers, such as AWS, Azure, and Google, have been integrating live VM and container migration~\cite{redhat-criu,google-e2,verma2015large,borg-criu,ruprecht2018vm} for the purposes, such as higher priority task preemption, kernel and firmware software updates, hardware updates, and reallocation for performance and availability. For example, the Google cluster manager Borg controls all computing tasks and container clusters of up to tens of thousands of physical machines. In Google production fleets, a lower bound of 1,000,000 migrations can be performed monthly~\cite{ruprecht2018vm}. 
These show the critical importance of migration management techniques in dynamic resource reallocation.

\begin{figure}[t]
	\centering
	\includegraphics[width=\linewidth]{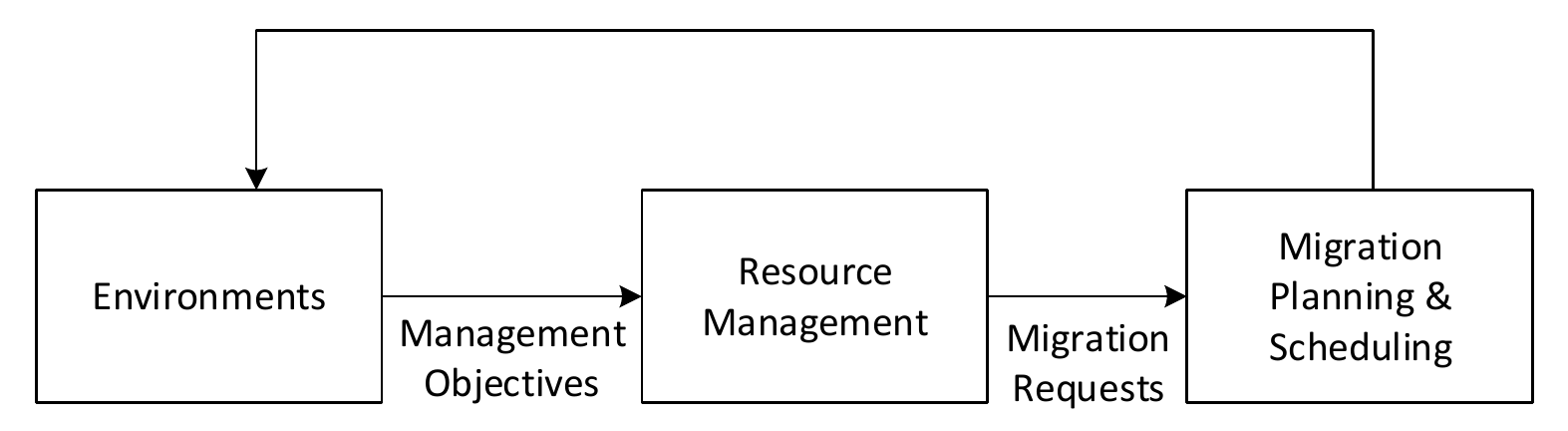}
	\caption{A general migration management framework}
	\label{fig: management-flow}
\end{figure}
Figure~\ref{fig: management-flow} illustrates the general migration management workflow. Based on the various objectives, the resource management algorithms
~\cite{verma2008pmapper, wood2009sandpiper, beloglazov2012optimal, mann2012, xu2014, bi2015application, wu2016energy,peng2019mobility,jiang2020network} 
find the optimal placement by generating multiple live migrations. With the generated multiple migration requests, the migration planning and scheduling algorithm~\cite{ghorbani2012, bari2014cqncr, wang2017virtual,he2021sla} optimizes the performance of multiple migrations, such as total and individual migration time and downtime, while minimizing the migration cost and overheads, such as migration impact on application QoS. On the other hand, the computing and networking resources are reallocated and affected by multiple migrations.

As a resource-intensive operation, live migration consumes both computing and networking resources when transmitting the memory dirty pages from the source to the destination host. It puts stress on both the migrating services and other services in the cloud data centers. Thus, it is crucial to minimize migration interference during dynamic resource management. There are continuous efforts to take migration overheads into consideration during the dynamic resource management \cite{verma2008pmapper, wood2009sandpiper, beloglazov2012optimal, xu2014}. Currently, most migration cost models consider overheads of single migration \cite{akoush2010, jo2017machine, he2019performance}, such as migration time (single execution time), downtime, transferred data with respect to the size of memory, dirty page rate, data compression rate and available bandwidth while allowing multiple migrations in dynamic resource management. For the migration selection, existing resource management algorithms utilize the linear cost model of single migration to minimize the overheads. Then, with the migration requests generated as the input, multiple migration planning and scheduling algorithms~\cite{ghorbani2012, bari2014cqncr, wang2017virtual, he2021sla} decide the sequence of migration requests to achieve the maximal scheduling performance.

There are obvious gaps regarding the multiple migration performance between the existing dynamic resource management policies, the migration cost model and the multiple migration scheduling. 
The total migration time, the time interval between the start of the first migration and the end of the last migration, is the convergence time for the resource management solution.
Overall, the real-time demands for live migration should be met by improving the performance in total migration time. For example, with the nature of highly variable workloads, SLA violations will occur as the resource demand surpasses the provisioned amount. In this case, a faster live migration convergence equals to less SLA violations. 

Resource dependency between two migrations, such as sharing source and destination hosts or network paths, can largely affect the performance of multiple migration scheduling. With the network as a bottleneck, two resource-dependent migrations can only be scheduled sequentially, while independent ones scheduled concurrently \cite{bari2014cqncr, wang2017virtual, he2019performance}.
If large amount of resource dependencies among migrations are generated by dynamic resource management, the performance of multiple migration scheduling will suffer a significant degradation. 
Since single migration overheads are only related to one migration,
it is critical to consider multiple migration overheads in order to generate migration requests with less resource dependencies.

Therefore, we incorporate the resource dependency of multiple migrations into the cost model to bridge the gaps. Based on the maximal cliques and independent sets of the dependency graph of potential migrations, we propose a concurrency-aware migration (CAMIG) selection strategy for migrating VMs and destination hosts of the dynamic resource management.
The \textbf{contributions} of this paper are summarized as follows:
\begin{itemize}
	\item We propose and model the multiple migration selection problem to minimize interference due to resource dependency among multiple migrations while achieving the objective of dynamic resource management.
	\item We introduce the resource dependency graph to model migration concurrency.
	\item We propose a flexible concurrency-aware migration selection strategy for dynamic resource management.
	\item We conduct extensive experiments in an event-driven simulation to show the performance improvement in terms of total migration time in correspondence with resource management objective.
\end{itemize}

The rest of the paper is organized as follows. Related works of migration cost management and multiple migration scheduling are reviewed in Section \ref{section: related}. The system framework and migration overheads are discussed in Section \ref{sect: 3}. The problem model is described in Section \ref{section: problem-model}. Section \ref{section: camig} proposes the concurrency-aware migration selection algorithm. Section \ref{section: evaluation} compares our proposed algorithm with other dynamic resource management algorithms in load-balancing and energy-saving scenarios. Finally, Section \ref{section: conclusion} summarizes the paper.

\section{Related Work} \label{section: related}
\begin{table*}[ht]
	\centering
	\caption{Comparison of approaches on dynamic resource management through live migration}
	\resizebox{\linewidth}{!}{
		\begin{tabular}{|l|lclll|}
			\hline
			algorithm							& resource management	&	single migration overhead	& dependency aware			&	migration performance						&	migration scheduling \\
			\hline
			\hline
			FFD \cite{verma2008pmapper}			&	load/energy			&	memory size			&  -							&   sum of migration cost						&	-					\\
			HARMONY \cite{singh2008server}		&   load				&	CPU, network		&  -							&	single exe. time							&	one-by-one					\\
			Sandpiper \cite{wood2009sandpiper}	&   load				&	memory size			&  -							&	single exe. time, migration number			&	-					\\
			Xiao et al. \cite{xiao2012dynamic}	&	load/energy			&	migration number	&  -							&	migration number							&	-					\\
			lrmmt \cite{beloglazov2012optimal}	&	load/energy			&	memory size			&  -							&	migration number							&	-					\\
			iAware \cite{xu2014}				&	flexible			&	single exe. time, computing	& - 		&	sum of normalized cost						&	one-by-one			\\			
			Our work (CAMIG)	  				&	flexible			&	migration model, computing		& computing, network sharing		&	total mig. time, downtime					&	multiple scheduling	\\	
			\hline	
	\end{tabular}}
	
	\label{tb: related-work}
\end{table*}
Many dynamic resource management solutions utilize live migration as a tool to achieve objectives, such as load-balancing~\cite{verma2008pmapper,singh2008server,wood2009sandpiper,mann2012}, energy efficiency~\cite{bi2015application,wu2016energy}, network delay~\cite{peng2019mobility}, and communication cost~\cite{jiang2020network}. Among these solutions, some resource management algorithms consider a linear model of the total migration overheads as the sum of individual migration overhead~\cite{verma2008pmapper,singh2008server,wood2009sandpiper,mann2012,beloglazov2012optimal,xu2014,wu2016energy}. 
However, existing research only considers the objectives of resource management while neglecting the multiple migration overheads and migration scheduling performance. 
Generally, during the dynamic resource management, there are three steps to generate migration requests: source host selection; VM selection; and destination host selection.
The overhead or interference model of single migration \cite{akoush2010, jo2017machine, he2019performance} is considered during the VM and destination selections.

For the VM and destination host selection, many dynamic resource management policies consider single migration overheads in terms of the memory size of migrating VM, single migration time, and the impact of one migration on other VMs located in the source or destination host, such as CPU, bandwidth of host network interface, and application bandwidth.
In the load balancing scenario, Verma et al. \cite{verma2008pmapper} estimated the migration cost based on the deduction of application throughput. It selects the smallest memory size VMs from the over-utilized hosts and assigns them to the under-utilized hosts in the First Fit Decreasing (FFD) order. Singh et al. \cite{singh2008server} proposed a multi-layer virtualization system HARMONY. It migrates VMs and data from hotspots on servers, network devices, and storage nodes. The load balancing algorithm is a variant of Toyoda multi-dimensional knapsack problem based on the evenness indicator Extended Vector Product (EVP). It considers the single live migration impact on application performance based on CPU congestion and network overheads. Wood et al. \cite{wood2009sandpiper} proposed  the load balancing algorithm Sandpiper that selects the smallest memory size VM from one of the most overloaded hosts to minimize the migration overheads. Mann et al. \cite{mann2012} focused on the VM and destination selection for the load balance of application network flows by considering the single migration cost model based on the dirty page rate, memory size, and available bandwidth.

In the energy-saving scenario, Xiao et al. \cite{xiao2012dynamic} investigated dynamic resource allocation through live migration. The proposed algorithm avoids the over-subscription while satisfying the resource needs of all VMs based on exponentially weighted moving average to predict the future loads. It also minimizes the physical machines regarding the energy consumption. Similarly, LR-MMT \cite{beloglazov2012optimal} focused on energy saving with local regression based on history utilization to avoid over-subscription. It chooses the least memory size VM from the over-utilized host and the most-energy saving destination.  Wu et al. ~\cite{wu2016energy} also studied the same problem of maximizing the power saving through VM consolidation by limiting individual migration costs. 
With the input of candidate VMs and destinations provided by other resource management algorithms, iAware \cite{xu2014} is a migration selector minimizing the single migration cost in terms of single migration execution time and host co-location interference. It considers dirty page rate, memory size, and available bandwidth for the single migration time. 
They argue that co-location interference from a single live migration on other VMs in the destination host in terms of performance degradation is linear to the number of VMs hosted by a physical machine in Xen.
However, it only considers one-by-one migration scheduling.

Taking the migration task list generated by resource algorithms as input, migration scheduling algorithms focus on minimizing the migration time by efficiently scheduling them. To find a possible sequence of migrations, one-by-one scheduling \cite{ghorbani2012} focused on avoiding the deadlock on the available resource of physical hosts. The multiple migration planning and scheduling algorithms \cite{bari2014cqncr, wang2017virtual, he2021sla} focused on the migration performance in terms of minimizing the total migration time by scheduling given migration tasks concurrently when necessary.
Table \ref{tb: related-work} summarizes representative related works and the proposed generic solution for existing dynamic resource management algorithms in management target, migration overhead, interference, migration performance, and migration scheduling method.

\section{Live Migration in Resource Management} \label{sect: 3}
We first introduce the background of live migration management including system overview and single cost model. Then, we discuss the resource dependency problem.
\subsection{System Overview} \label{section: overview}
\begin{figure}[t!]
	\centering
	\includegraphics[width=0.80\linewidth]{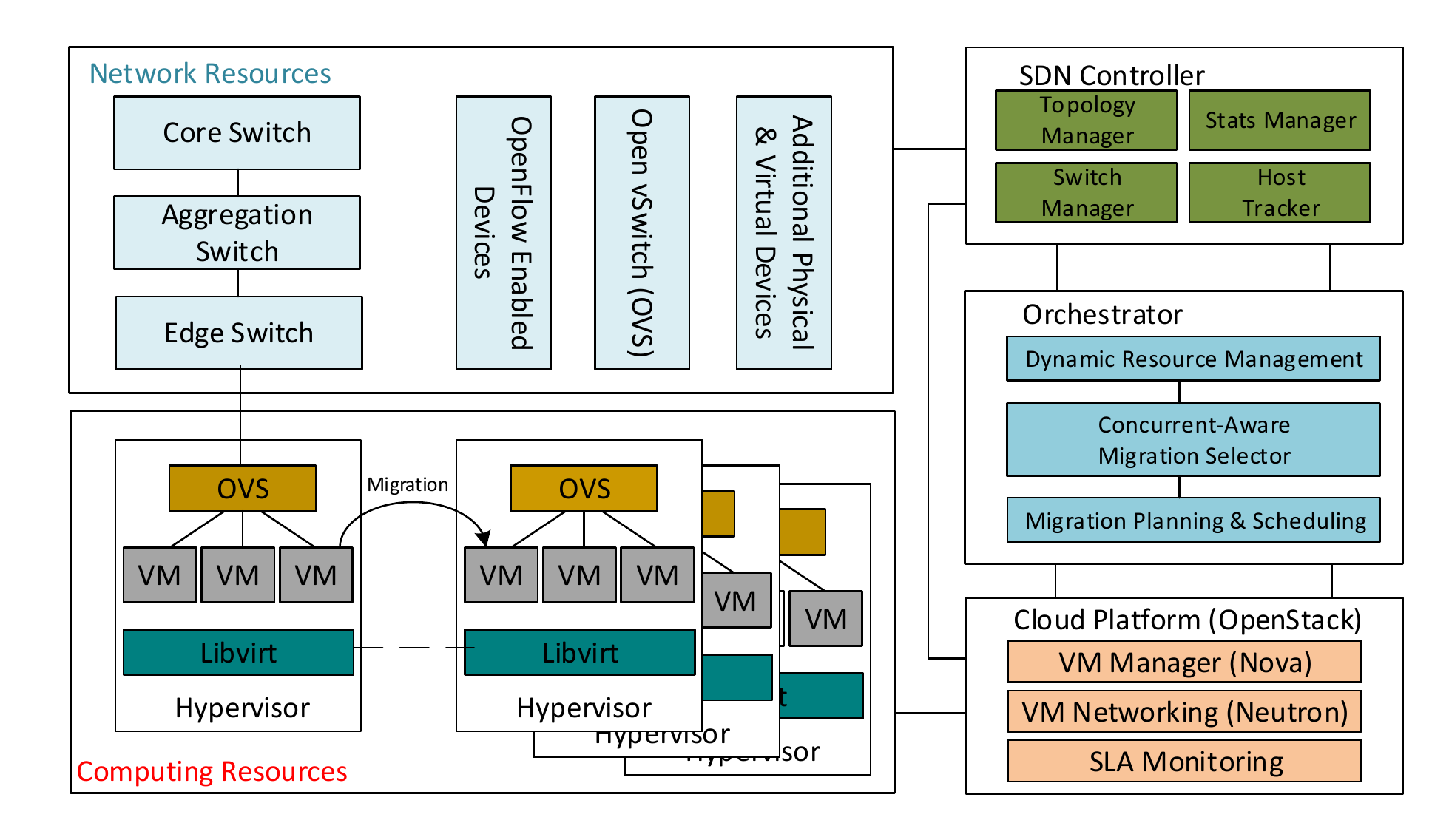}
	\caption{System Overview}
	\label{fig: overview}
\end{figure}
By integrating Software-Defined Networks (SDN) \cite{mckeown2008openflow}, the SDN-enabled cloud data centers have a centralized solution for the monitoring, planning, and scheduling of virtualized computing and networking resources \cite{son2017taxonomy}.
Fig.~\ref{fig: overview} illustrates the migration framework in the orchestration layer. The dynamic resource manager integrated with migration selector and multiple migration scheduler based on both monitoring computing resource and network resources. 
VMs are hosted on physical machines to provide various cloud services. 
Computing resources are controlled by VM Manager (VMM), such as OpenStack Nova, while the networking resources (such as available bandwidth and routing) are managed by the SDN controller and VM Networking Service, such as OpenStack Neutron, in a centralized way.
The SDN controller can dynamically manage the routing for migration elephant flows to avoid the congestion and alleviate the impact on cloud services. We can predict the cost of live migration by the available bandwidth between the source and destination hosts.

\subsection{Single Migration Cost Model} \label{section: mig-model}

To better understand the impact of multiple migrations on performance in dynamic resource management settings, we first introduce the mathematical model of a single live migration~\cite{he2019performance}. Live migration can be categorized into two types: post-copy and pre-copy migration. Since the pre-copy migration~\cite{clark2005live} is the most widely used approach in hypervisors (KVM, VMWare, Xen, etc), we consider it as the base model. During the pre-copy live migration for VMs or Containers, the hypervisor or the Checkpoint/Restore agent in the userspace (CRIU) \cite{criu}  iteratively copies the dirty memory pages in the previous transmission interval from the source host to the destination host.

The most important aspect of single migration overheads is the migration time or the single migration execution time. According to the live migration process \cite{clark2005live}, the pre-copy live migration consists of eight phases (see Fig. \ref{fig: live-migration}): pre-migration, initialization, reservation, iterative memory copy, stop-and-copy, commitment, and post-migration. Thus, live migration consumes both computing resources (pre-/post-migration overheads) and networking resources (memory copy and dirty page transmission) \cite{he2019performance}. The total single migration time $T_{mig}$ can be categorized into three parts: pre-migration computing overheads, memory-copy networking overheads, and post-migration computing overheads: 
\begin{equation}\label{eq:total-com-mem}
T_{mig} = T_{pre} + T_{mem} + T_{post}
\end{equation}

\begin{figure}[t!]
	\centering
	\includegraphics[width=\linewidth]{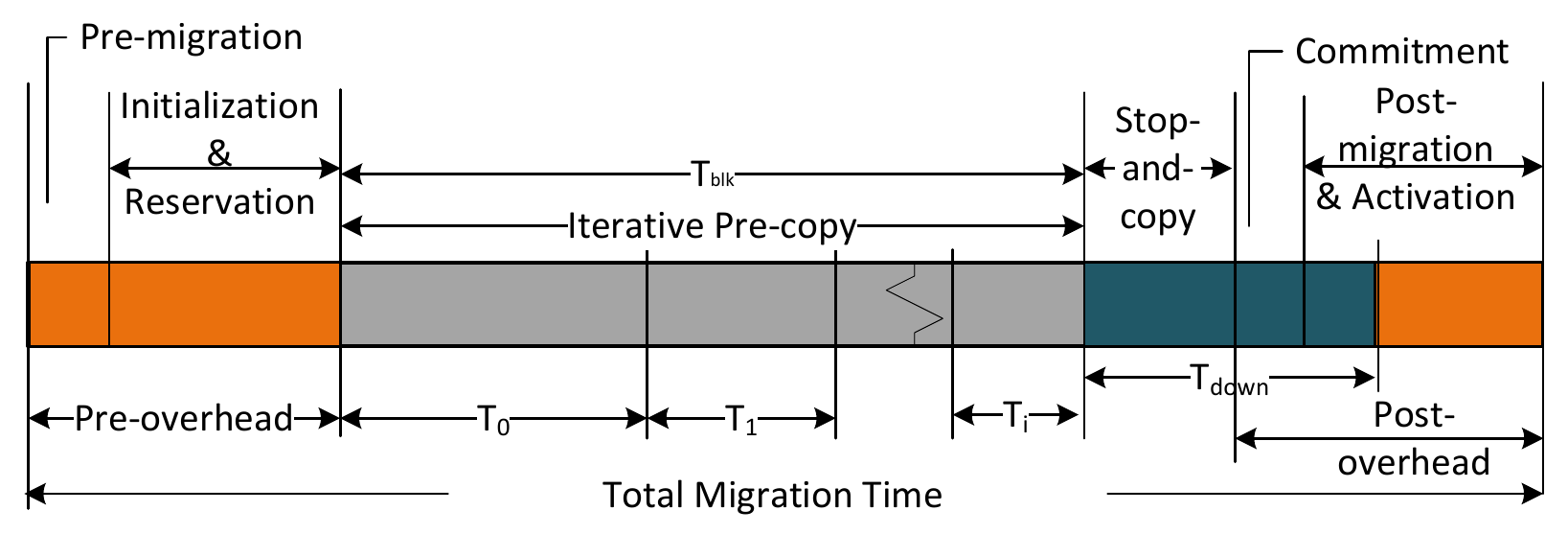}
	\caption{Pre-copy Live Migration}
	\label{fig: live-migration}
\end{figure}

Based on the iterative pre-copy illustrated in Fig. \ref{fig: live-migration}, the migration performance in terms of memory-copy can be represented as \cite{he2019performance}:
\begin{equation} \label{eq: mig-time}
{T_{mem}} = \frac{{\rho  \cdot {Mem}}}{L} \cdot \frac{{1 - {\sigma ^{i + 1}}}}{{1 - \sigma }}
\end{equation}
\begin{equation}\label{eq:round-number}
i = \min \left( {\left\lceil {{{\log }_\sigma }\frac{{{V_{thd}}}}{M}} \right\rceil ,\Theta } \right)
\end{equation}
where the ratio $\sigma  = \rho  \cdot {R \mathord{\left/ {\vphantom {R L}} \right. \kern-\nulldelimiterspace} L}$, $\rho$ is the compression rate of dirty memory, $Mem$ is memory size, $L$ is available bandwidth, $R$ is dirty page rate, $i$ is the total migration round, $\Theta$ denotes the maximum allowed number of iteration rounds, $V_{thd} = T_{dthd} \cdot L_{i-1}$ is the remaining dirty pages need to be transferred in the stop-and-copy phase, and $T_{dthd}$ is the configured downtime threshold.

\subsection{Resource Dependency}
Not only the overheads of the single migration but also resource dependencies among multiple migrations can heavily affect the performance of dynamic resource management.

For dynamic resource management policies, there are three selection steps: (1) selection of source physical hosts that need to be adjusted based on the management objective; (2) selection of VM(s) which need to be migrated from the selected host(s); and (3) selection of destination hosts of live VM migrations among potential candidates. With the input of candidate VMs and available destination hosts, different combinations of source and destination can achieve the same objective of dynamic resource management. However, there is a huge difference between these combinations in the scheduling performance of multiple migrations due to resource dependencies among migrations. If sharing the same source or destination hosts, or part of the network routing, two live migrations are resource-dependent. 

Two resource-dependent migrations can not be scheduled at the same time \cite{bari2014cqncr, he2019performance}. Because, according to equation (\ref{eq: mig-time}), larger bandwidth allocation means a smaller migration execution time and downtime. Thus, the networking resources are the bottlenecks which need to be optimized during the multiple migrations. 
For example, we have a number of migrations partially or entirely sharing network paths. Based on equation (\ref{eq: mig-time}), if scheduled at the same time, experimental results \cite{he2019performance} show that the total migration time will be more than the sum of single execution time.
Thus, sequential scheduling of dependent migrations is the most efficient way to optimize the migration performance\cite{bari2014cqncr, he2019performance}. Meanwhile, migrations which are resource-independent can be scheduled concurrently to reduce the total migration time. Therefore, it is essential to exclusively allocate one network path to only one migration until it is finished to achieve the optimal total migration time, average execution time, and downtime. 

\subsection{Illustrative Example} \label{sect: example}
\begin{figure}[t]
	\centering
	\begin{subfigure}{.45\linewidth}
		\centering
		\includegraphics[width=\linewidth]{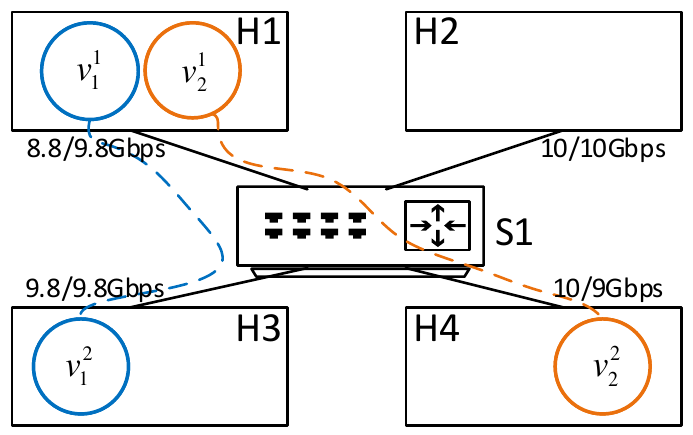}
		\caption{}
		\label{fig: example}
	\end{subfigure}%
	\hfil
	\begin{subfigure}{.37\linewidth}
		\centering
		\includegraphics[width=\linewidth]{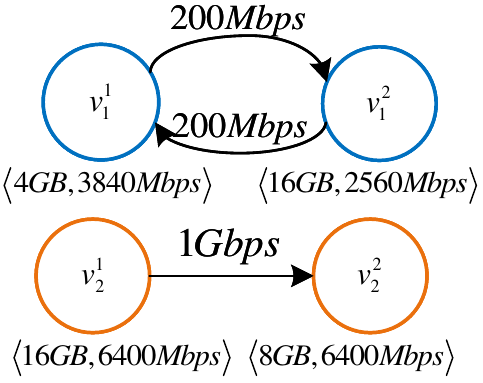}
		\caption{}
		\label{fig: example-vms}
	\end{subfigure}
	\caption{Scenario of Resource Dependencies during Migration Selections: (a) Initial Placement and (b) Virtual Connections between VMs with Memory Size and Dirty Page Rate}
\end{figure}

Fig. \ref{fig: example} shows the initial VM placement of the illustrative example along with the resource dependency among possible migration selections.  Fig. \ref{fig: example-vms} illustrates the virtual connections between VMs and the memory size (GB) and dirty page rate (Mbps) for each. Moreover, the threshold of iteration rounds is 30 and downtime threshold is 0.5 seconds.
The objective of the management policy is to reduce the communication cost by VM consolidation. There are several potential migration combinations which can fulfill the objective: 
	M1: $v_1^1: H1 \to H3$ and $v_2^1: H1 \to H4$; 
	M2: $v_1^1: H1 \to H3$ and $v_2^2: H4 \to H1$; 
	M3: $v_1^2: H3 \to H1$ and $v_2^1: H1 \to H4$; and 
	M4: $v_1^2: H3 \to H1$ and $v_2^2: H4 \to H1$.
We can schedule two resource-independent migrations concurrently (M2 and M3). On the other hand, one migration can only be scheduled in sequence after the completion of another dependent migration (M1 and M4). 

We use Mininet \cite{lantz2010network} to emulate the iterative network transmission of the live migration. The execution time for each potential migration of $v_1^1$, $v_1^2$, $v_2^1$, and $v_2^2$ based on the available bandwidth is 6.2791, 15.0889, 29.1980, and 12.5143 seconds, respectively. The total migration time of combination M1-M4 is 34.8858, 12.4334, 28.4711, and 27.6032 seconds. Moreover, when the service network and migration (control) network are running separately \cite{tsakalozos2017live}, the available bandwidth for each live migration is the same (10 Gbps). Based on multiple migration planning and scheduling algorithms~\cite{bari2014cqncr, wang2017virtual, he2021sla}, the total migration time of four different combinations M1-M4 is 28.1936, 12.1227, 22.6056, and 26.8893 seconds, respectively. Comparing M2 with M1 and M4, since there is no resource-dependent migration in M2, the total migration time is significantly shorter.  Comparing M2 with M3, although there is no network resource sharing in both combinations, the single live migration overheads of M2 is smaller due to the memory size, dirty page rate, and available bandwidth. Summarily, although all the potential combinations can achieve the desired objective, the scheduling performance of multiple migrations varies considerably. Thus, it is essential to minimize both resource dependencies among migration requests and single live migration overheads during dynamic resource management.

\section{Problem Modeling} \label{section: problem-model}
In this section, we model the problem of multiple migration selection to minimize the migration dependency while achieving the objective of dynamic resource management as a Mixed Integer Programming (MIP) problem. 

In the model, $H$ is the set of all candidate destination physical hosts $h \in H$ while $N$ denotes the set of candidate VMs $i \in N$ for the migration. $H_i$ is the set of candidate hosts for VM $i$.
Let binary variable ${y_{(i,h)}} \in \left\{ {1,0} \right\}$ indicate both initial and final placement of VM $i$ in host h. When the VM $i$ is in the initial host $p_i$, ${y_{(i,p_i)}} = 1$. When VM $i$ is in the host $h$ in the final placement, ${y_{(i,h)}} = 1$. Otherwise, ${y_{(i,h)}} = 0$.
Let the binary variable ${{x_{(i,h)}} \in \left\{ {1,0} \right\}}$ indicate whether VM $i$ is in the host $i$ in the final placement. In other words, if VM $i$ is migrated to host $h$, then $x_{(i,h)}=1$ and $h!=p_i$. If VM $i$ is not migrated, then $x_{(i,h)}=1$ and $h=p_i$. Otherwise, $x_{(i,h)}=0$ which indicates that VM $i$ is not in host $h$ in the final placement determined by the dynamic resource management policy.

To generalize the problem, we can omit the VM index $i$ for $h \in H_i$ by adding extra constraints to ${x_{(i,h)}}$ when some destination hosts are not available for the specific VM $i$:
\begin{equation} \label{eq: dest}
\begin{array}{*{20}{c}}
{{x_{(i,{{\overline h }_i})}} = 0}&{\forall {{\overline h }_i} \in {{\overline H }_i} = H\backslash {H_i}}
\end{array}
\end{equation}
where ${{\overline h }_i}$ indicates the unavailable host for VM $i$. 

The migration execution time $t_i^h$ of ${{x_{(i,h)}}} =1, h!=p_i$ can be calculated according to equations (\ref{eq:total-com-mem})-(\ref{eq:round-number}). Furthermore, we normalize the migration execution time based on the largest and smallest execution time among the different source and destination pairs for every VMs.

As there can be only one destination and the VM must be allocated in one and only one host at the same time, we add the following constraints to the binary variable ${{x_{(i,h)}}}$:
\begin{equation} \label{eq: one-mapping}
\begin{array}{*{20}{c}}
{\sum\limits_{h \in H} {{x_{(i,h)}} = 1} }&{\forall i \in N}
\end{array}
\end{equation}

The VM $i$ can only be migrated from source host of the initial placement $h_s = p_i$ where ${y_{(i,p_i)}} = 1$ to the destination host of the final placement $h_d$ that ${y_{(i,h_d)}} = 1$, ${x_{(i,h_d)}} = 1$ and ${x_{(i,p_i)}} = 0$ or not be migrated at all ${x_{(i,h_d)}} = 1, h_d=p_i$. 
Thus, we have the constraints expression as follows:
\begin{equation} \label{eq: mig-num}
\begin{array}{*{20}{c}}
{{x_{(i,h)}} - {y_{(i,h)}} \le 0}&{\forall i,h \in N \times H}
\end{array}
\end{equation}

Constraints of the placement binary variable $y_{(i,h)}$ are:
\begin{equation}
\begin{array}{*{20}{c}}
{1 \le \sum\limits_{h \in H} {{y_{(i,h)}} \le 2} }&{\forall i \in N}
\end{array}
\end{equation}
where $\sum\limits_{h \in H} {{y_{(i,h)}} = 2}$, when VM $i$ is migrated to other host in the final placement. $\sum\limits_{h \in H} {{y_{(i,h)}} = 1}$, when VM $i$ is still in host $p_i$ in the final placement.

Let ${{z_{(i,j,{h_1}, {h_2})}}}$ denote the binary variable indicating whether VM $i$ and $j$ are migrated to destination $h_1$ and $h_2$:
\begin{equation} \label{eq: mig-dep}
\begin{array}{*{20}{c}}
{{z_{(i,j,{h_1},{h_2})}} \in \left\{ {1,0} \right\}}&{\forall i,j \in N,{h_1},{h_2} \in H}
\end{array}
\end{equation}
where $z_{(i,j,h1,h2)}=1$, if $ y_{(i,h1)}=1$, $y_{(j,h2)}=1$ and $p_i!=h1$, $ p_j!=h2$. Otherwise, $z_{(i,j,h1,h2)}=0$.

There is a resource dependency graph $G_{dep}$ for all possible migrations. Let $v_{s,d}$ denote a migration with source host $s$ and destination host $d$. 
If node $v_{p_i,h1}$ and $v_{p_j,h2}$ are connected in graph $G_{dep}$, then edge $e_{(i,j,h1,h2)}=1$. This indicates that potential migrations of VM $i$ from host $p_i$ to $h1$ and VM $j$ from host $p_j$ to $h2$ are resource-dependent which can only be scheduled in a sequential manner.
Thus, the resource dependency between two potential migrations can be represented as:
\begin{equation}
e_{(i,j,h1,h2)} \cdot z_{(i,j,h1,h2)}
\end{equation}

Let ${O_{init}}$ and ${O_{tar}}$ denote the initial score and target score of dynamic resource management and $\varepsilon$ represent the tolerant value for accepted range. Let $O({x_{(i,h)}})$ denote the objective score achieved after all migrations based on $x_{(i,h)}$ indicator. Thus, the constraints of final placement for dynamic resource management can be represented as:
\begin{equation} \label{eq: target}
\begin{array}{*{20}{c}}
{\left| {O\left( {{x_{(i,h)}}} \right) - {O_{tar}}} \right| \le \varepsilon }&{\forall (i,h) \in N \times H}
\end{array}
\end{equation}
In practice, we can replace (\ref{eq: target}) for a specific placement score function. For example, in load balancing policies, let $w_i$ and $w_j$ denote the load of VM $i$ and $j$. We can represent the constraints of dynamic resource target for the final placement as:
\begin{equation} \label{eq: lb}
{\left| {\sum\limits_{i \in N} {{x_{(i,{h_1})}} \cdot {w_i} - \sum\limits_{j \in N} {{x_{(i,{h_2})}}} }  \cdot {w_j}} \right| \le \varepsilon^{'} }
\end{equation}
where ${\forall ({h_1},{h_2}) \in H \times H:{h_1} \ne {h_2}}$ and $\varepsilon^{'}$ is the tolerant value among the physical hosts.

In addition, let $C_{\left( {Mem,Core,Disk,Work} \right)}^h$ = $(1,1,1,1)$ denote the normalized computing resource capacity of physical host $h$ for memory $Mem$, CPU $Core$, storage disk $Disk$, and total workload $Work$. Therefore, the constraints of computing resources, such as workload, can be represented by:
\begin{equation} \label{eq: res-cap}
\begin{array}{*{20}{c}}
{\sum\limits_{i \in N} {{x_{(i,h)}} \cdot {w_i} \le } C_{\left( {Work} \right)}^h}&{\forall h \in H}
\end{array}
\end{equation}

The single and multiple migration overheads, $Inter_{single}$ and $Inter_{multi}$, are calculated as:
\begin{equation}
Inte{r_{single}} = \sum\limits_{i \in N} {\left( {{y_{(i,{p_i})}} - {x_{(i,{p_i})}}} \right)\cdot t_i^h}
\end{equation}
\begin{equation}
Inte{r_{multi}} = \sum\limits_{i,j \in N,h1,h2 \in H} {\left( {t_i^{h1} + t_j^{h2}} \right) \cdot e \cdot z}
\end{equation}
where $e$ and $z$ omit the subscripts for a concise equation.

Therefore, the objective of the problem in terms of minimizing both single migration overheads and resource dependencies among multiple migration requests can be formulated as:
\begin{equation} \label{eq: model-objective}
\min({Inte{r_{single}} + Inte{r_{multi}}})
\end{equation}
subject to constraints (\ref{eq: dest}) - (\ref{eq: res-cap}).

The objective function contains two parts: the first objective is for the sum of single migration overhead, where ${t_i^h}$ indicates single migration time of VM $i$ from source host $p_i$ to destination host $h$. Note that, although only migration time is modeled, it can be extended to other interference, such as CPU congestions, heterogeneous links, bandwidth overheads on other applications, and the number of co-located VMs in the destination host. The second part is multiple migration overheads during multiple migration scheduling. Namely, it indicates how much overheads due to resource dependencies happened. The fewer dependencies in migration requests with less individual overheads, the greater possibility of larger concurrent migration groups during scheduling, which results in a shorter total migration time.

\section{Concurrency-Aware Selection} \label{section: camig}
Solving the MIP model in Equation (\ref{eq: model-objective}) is NP-hard, it is not practical to use MIP solver to get the solution. In this section, we introduce the Concurrency-Aware Migration (CAMIG) selection algorithm for minimizing the resource dependencies and overheads among migrations during dynamic resource management.
Based on the three selection steps of resource management policy, CAMIG has the flexibility to integrate with existing algorithms. Provided that VMs are selected by the policy, CAMIG selects migration destinations to minimize resource dependency. Moreover, if only the management objective and source host selection criteria are given, CAMIG selects both VMs and migration destinations.

The rationale behind CAMIG is to select the migration with the least resource dependency and single migration overhead in each round with the currently selected migrations and minimize the dependency for the future one based on maximal cliques and independent sets of the resource dependency graph. 
Graph theory concepts, such as maximal cliques and independent sets, are explained in Section \ref{sect: cliq-ind}.
There are mainly three steps:
(1) build the migration dependency graph;
(2) get all maximal cliques and independent sets of a migration from the dependency graph; and
(3) calculate the single migration interference and migration concurrency metric (MIGC) of candidate migrations.

\subsection{Migration Dependency Graph Build}
\begin{figure}[t]
	\centering
	\includegraphics[width=0.8\linewidth]{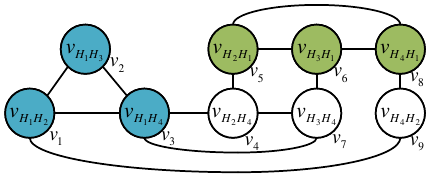}
	\caption{A Resource Dependency Graph with Two of its Maximal Cliques Marked by Color}
	\label{fig: depgraph-example}
\end{figure}
We first explain how to generate the resource dependency graph $G_{dep}$ based on the potential migrating VMs and destinations.
For the undirected graph $G_{dep} =(V, E)$, let $v$ ($v \in V$) be the source-destination pair (src-dst) node or vertex representing one potential migration. Migrations with same src-dst node are categorized in list $M\left( {{v_{sd}}} \right)$. Let $e(v,u) \in E$ be the dependency between two migrations with src-dst node $v$ and $u$. 
As shown in Algorithm \ref{alg: depgraph}, with the input of potential migrating VMs and corresponding destination candidates $H_{i}$, we first add src-dst nodes and classify potential migrations into the corresponding node in $M\left( {{v_{sd}}} \right)$. Then, we add edges into $G_{dep}$ based on the source and destination of each node. 
Fig. \ref{fig: depgraph-example} demonstrates an illustrative example of resource dependency graph based on a given list of potential migrations ($v_1$ to $v_9$) in a specific dynamic resource management which involves 9 src-dst pairs in the same physical network topology shown in Fig. \ref{fig: example} (four hosts connected through one switch).
Each vertex $v_{H_{s}H_{d}}$ indicates the pair of source and destination host for a group of potential migrations. For the sake of conciseness, we use $v_1$ to $v_9$ to represent node $v_{H_{1}H_{2}}$ to $v_{H_{4}H_{2}}$.

\begin{algorithm}[t!]
	\caption{Create $G_{dep}$ and $v_{sd}$ queues}\label{alg: depgraph}
	\KwIn{potential VM $i \in N$, Destinations $\{H_{i}\}$}
	\KwResult{migration depGraph $G_{dep}$, $\{M(v_{sd})\}$}
	\SetKwFunction{depG}{\textsc{Is}Dependent}
	\SetKwFunction{addN}{\textsc{Add}Node}
	\SetKwFunction{addE}{\textsc{Add}Edge}
	\ForEach{$i \in N$}{
		$s \gets p_i$;\\
		\ForEach{$d \in H_{i}$}{
			\addN($G_{dep}$, $v_{sd}$);\\
			$M(v_{sd}) \gets M(v_{sd}) \cup i$;\\
		}
	}
	\ForEach{$v \in V(G_{dep})$}{
		\ForEach{$u \in V(G_{dep})$}{
			\If{$v != u$}{
				\If{\depG($u$,$v$)}{
					\addE($G_{dep}$, ($u,v$));\\
				}
			}
		}
	}
	\KwRet{$G_{dep}$ , $\{M(v_{sd})\}$}
\end{algorithm}

Regardless of the number of potential migrations, the scale of $G_{dep}$ only depends on the source and destination hosts involved. 
Given a list of migrations $M$ = $\{m_0, m_1, ..., m_n\}$, the dependency graph $G(M)$ of $M$ can be constructed as $G(M) = (V, E)$.  As migrations with the same source and destination are always resource-dependent, we categorize migrations into different lists of src-dst pair $v$. Then, all migrations can be represented as $\{M(v_{sd})\} = \{M(v_0),...,M(v_{|V|})\}$. The size of node $|V|$ in the migration dependency graph will be the total combination of source and destination hosts. Through this pre-processing, the total nodes of $G_{dep}$ can be reduced from as many as the potential migrations $|M|$ to the migration pair participated $|V|$. Therefore, the upper-bound of total nodes in graph $G_{dep}(M)$ is $|H_{src}| \cdot |H_{dst}|$. $H_{src}$ and $H_{dst}$ are the number of potential source and destination hosts, respectively.

Note that the dependency graph supports the multiple routing transmission and dynamic migration routing based on the current network status. 
In certain data center networks, multi-path transmission and multiple network interfaces of physical hosts are supported. Thus, the vertex $v_{sd}^P$ in $G_{dep}$ can be extended to indicate the network paths $P_{sd}$ for migrations from the specified network interfaces set $s$ of source host to interfaces set $d$ of destination host. Let $u(P)$ indicate the available bandwidth of network paths $P$. Given two pairs of src-dst interfaces set  $(s_j,d_j)$ and $(s_k, d_k)$ and corresponding network paths $P_j$ and $P_k$, two vertices $v_j$ and $v_k$ are resource-independent, when statement (\ref{eq: independent-statement}) are true and ${s_k} \cap {s_j} = \emptyset$ and ${d_k} \cap {d_j} = \emptyset$:
\begin{equation} \label{eq: independent-statement}
\begin{array}{*{20}{c}}
{u\left( {{P_j}} \right) - u\left( {{P_j} \cap {P_k}} \right) \ge \min \left( {u\left( {{P_j}} \right),NC_s^j,NC_d^j} \right) \wedge}\\
{u\left( {{P_k}} \right) - u\left( {{P_k} \cap {P_j}} \right) \ge \min \left( {u\left( {{P_k}} \right),NC_s^k,NC_d^k} \right)}
\end{array}
\end{equation}
where $(NC_s^j,NC_d^j)$ and $(NC_s^k, NC_d^k)$ indicate the network capacity of interface set and $u\left( {{P_j}} \right)$ and $u\left( {{P_k}} \right)$ indicate the available bandwidth of network paths.
Otherwise, the two vertices are resource-dependent. The upper bound of total nodes in $G_{dep}$ is the total number of $P_{sd}$.

\subsection{Maximal Cliques and Independent Sets} \label{sect: cliq-ind}
\begin{figure}[t]
	\centering
	\includegraphics[width=\linewidth]{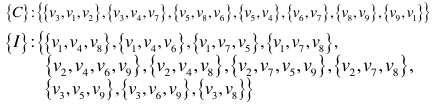}
	\caption{All Maximal Cliques and MISs of $G_{dep}$ in Fig. \ref{fig: depgraph-example}}
	\label{fig: depgraph-example-clique}
\end{figure}

Before discussing how to get maximal cliques and maximal independent sets (MISs) which include a certain node $v$,
we first review some basic concepts, such as clique, independent set, and degeneracy. A clique is a subset of vertices of an undirected graph $G$ such that every two distinct vertices in the subset are adjacent \cite{bron1973algorithm}. The maximal clique is a clique that cannot be extended by including one more adjacent vertex. An independent set of a graph $G$ is the opposite of a clique that no two nodes in the set are adjacent.
Fig. \ref{fig: depgraph-example-clique} shows all maximal cliques and MISs of the $G_{dep}$ (Fig. \ref{fig: depgraph-example}). For example, $\left\{ {{v_3},{v_1},{v_2}} \right\}$ is one of its maximal cliques and $\left\{ {{v_2},{v_7},{v_5},{v_9}} \right\}$ is one of its MISs.
The problems of finding all maximal independent sets and cliques are complementary and NP-hard \cite{bron1973algorithm, lawler1980generating}. Finding all maximal independent sets of a graph is equal to finding all maximal cliques of its complement graph \cite{tomita2006worst}. As a robust metric to indicate graph density or spareness, degeneracy of a graph $G$ is the smallest value $d$ such that every nonempty subgraph of $G$ contains a vertex of degree at most $d$ \cite{lick1970k}.

A clique of $G_{dep}$ is a set of src-dst nodes, where migrations with these nodes can not be scheduled at the same time. In contrast, the migrations from the src-dst nodes within an independent set can be scheduled concurrently.
To check and evaluate the resource dependency or concurrency of each migration with src-dst pair node $v$, we need to generate all maximal cliques $\{ {C^v} \}$ and MISs $\{ {I^v}\}$ of $G_{dep}$ including node $v$ . 
Let $\{C\}$ and $\{I\}$ be all maximal cliques and all maximal independent sets of $G_{dep}$, where $C \in \{C\}$ and $I \in \{I\}$ is one of the maximal cliques and MISs. Let ${C^v} \in \left\{ C \right\}$ and ${I^v} \in \left\{ I \right\}$ denote one of the maximal cliques and independent sets including node $v$. Then, $\left\{ {{C^v}} \right\} \subseteq \left\{ C \right\}$ and $\left\{ {{I^v}} \right\} \subseteq \left\{ I \right\}$. 

\begin{algorithm}[t!]
	\caption{Get All MISs of Node $v$ in $G_{dep}$}\label{alg: clique-indep-alg}	
	\KwIn{$G, \bar{G}, v$, node neighbors in $\bar{G}$ $\{\bar{G}_N(n)\}, n\in V$}
	\SetKwFunction{findC}{\textsc{Cliques}}
	\SetKwFunction{maxD}{\textsc{Get}NodeMaxDegree}
	\KwResult{All MISs $\{I^v\}$ of node $v$, $v \in V({G_{dep}})$}
	\SetKwProg{Fn}{Function}{:}{}
	\Fn{\findC(${\bar G}$, $cand$)}{
		$n \gets$ \maxD(${\bar G}$);\\
		\ForEach{$m \in cand - \bar{G}_N(n)$}{
			$del(m, cand)$;\\
			$I \gets I \cup \{m\}$;\\
			\eIf{${\bar G} \cap \bar{G}_N(m) = \emptyset$}{$\{I^v\} \gets \{I^v\} \cup I$;}{
				\If{$cand \cap adj[m] != \emptyset$}{
					\findC(${\bar G} \cap  \bar{G}_N(m)$, $cand \cap \bar{G}_N(m)$);
				}
			}
			$del(m, I)$;\\
		}
	}
	\textbf{End Function}\\
	${I^v} \gets \emptyset$; $I \gets \{v\}$;\\
	$cand \gets cand - G_N(v) - \{v\}$;\\
	$V(G) \gets V(G) - G_N(v) - \{v\}$;\\
	
	\Return \findC(${\bar G}$, $cand$);
\end{algorithm}

We propose an algorithm for listing $\{C^v\}$ and $\{I^v\}$ based on $\{C\}$ of dependency graph. 
For getting all maximal cliques $\{C\}$ of a graph,
the general-purpose algorithms for listing all maximal cliques \cite{cazals2008note, tomita2006worst} based on Bron-Kerbosch algorithm \cite{bron1973algorithm} take exponential time due to the maximum possible number of cliques. These general-purpose algorithms are not sensitive to the density of a graph. Therefore, parametrized by degeneracy, we use a variant algorithm Bron-Kerbosch Degeneracy \cite{eppstein2010listing} to generate all maximal cliques of the original resource-dependency graph without duplication. All maximal cliques are generated in the tree-like structure by employing the pruning methods with pivoting to allow quick backtrack during the search.  Based on the Bron-Kerbosch algorithm with pivoting, the Bron-Kerbosch Degeneracy uses a degeneracy ordering to order the sequence of recursive calls without pivoting at the outer level of the original Bron-Kerbosch algorithm \cite{eppstein2010listing}. Applied to a $n$-vertex graph with $d$ degeneracy, it lists all maximal cliques in time $O\left( {dn{3^{d/3}}} \right)$.

As shown in the dependency graph property analysis (Appendix A) and the time analysis in the performance evaluation (Section \ref{sect: time-analysis}), 
it is not practical to generate all maximal independent sets $\{I\}$ due to the density of the complement of $G_{dep}$. Thus, we propose a clique-based maximal independent set algorithm to calculate $\{I^v\}$. As shown in Algorithm \ref{alg: clique-indep-alg}, it fist excludes all adjacent nodes of $v$ in the resource dependency graph $G$. Then, it chooses node with maximum degree from each connected candidates of the remaining complement graph $\hat{G}$ recursively in a branch-and-bound method until there is no vertex left. 
Algorithm 2 can achieve the worst-case optimal time complexity of finding all MISs of a node $v$ as $O\left( 3^{m/3}\right)$~\cite{cazals2008note}, where $m = |V(G) - G_N(v)|-1$.

\subsection{Concurrency for Migration Candidates}
In this section, we introduce the migration concurrency metric (MIGC) to indicate the resource dependency level of a potential migration. It is based on the maximal cliques and independent sets of an src-dst pair node. Let $M_{mig}^x$ be the list of migrations have been selected currently. Let $M^x$ be the list of src-dst pair nodes $v_j$ of each migration $m_j \in M_{mig}^x$.
For the first round $x=0$, when the list of selected VM migration is empty, MIGC can be calculated as: 
\begin{equation}\label{eq: migc1}
MIG{C_v} = {{\kappa \cdot\max \left( {\left| {{C^v}} \right|} \right)} \mathord{\left/
		{\vphantom {{\kappa \cdot\max \left( {\left| {{C^v}} \right|} \right)} {\max \left( {\left| {{I^v}} \right|} \right)}}} \right.
		\kern-\nulldelimiterspace} {\max \left( {\left| {{I^v}} \right|} \right)}}
\end{equation}
where ${I^v} \in \left\{ {{I^v}} \right\}$ and ${C^v} \in \left\{ {{C^v}} \right\}$, $\kappa$ is the coefficient for the value normalization.
When $x>0$, the MIGC of migration with src-dst pair node $v$ in $G_{dep}$ can be represented as:
\begin{equation}\label{eq: migc}
MIGC_v^{{M^x}} = MIGCliq_v^{{M^x}} + {1 \mathord{\left/
		{\vphantom {1 {MIGInd_v^{{M^x}}}}} \right.
		\kern-\nulldelimiterspace} {MIGInd_v^{{M^x}}}}
\end{equation}
The migration independent score of the testing node $v$ regarding to the selected migration list can be calculated as:
\begin{equation} \label{eq: migind}
MIGInd_v^{{M^x}} = \frac{{\sum\limits_{{v_j} \in {M^x}} {\sum\limits_{{I^v} \in \{ {I^v}\} } {\left| {{v_j} \cap {I^v}} \right|} } }}{{\left| {\{ {I^v}\} } \right|\cdot\left| {{M^x}} \right|}}
\end{equation}
where ${\sum\limits_{{v_j} \in {M^x}} {\sum\limits_{{I^v} \in \{ {I^v}\} } {\left| {{v_j} \cap {I^v}} \right|} } }$ indicates how many times src-dst nodes $v_j$ of migration from the currently selected list $v_j \in M^x$ is shown in all MISs of the testing node $v$. ${{\left| {\{ {I^v}\} } \right|\cdot\left| {{M^x}} \right|}}$ is the product of the total number of $I^v$ and the number of selected migrations.

Similarly, the migration clique score for src-dst pair node $v$ according to the node list of currently selected migrations $M^x$ is represented as:
\begin{equation} \label{eq: migcliq}
MIGCliq_v^{{M^x}} = \frac{{\sum\limits_{{v_j} \in {M^x}} {\sum\limits_{{I^v} \in \{ {C^v}\} } {\left| {{v_j} \cap {C^v}} \right|} } }}{{\left| {\{ {C^v}\} } \right|\cdot\left| {{M^x}} \right|}}
\end{equation}
where the numerator part indicates how many times the src-dst pair nodes of currently selected migrations is included in the maximal cliques of the node $v$.

The range of the migration clique score and independent set score is $MIGCliq \in [0,1]$ and $MIGInd \in (0,1]$. The largest $MIGCliq$ is 1 when all src-dst pair nodes of selected migrations in $M$ shown in every maximal clique of the testing node. $MIGCliq$ is 0 when there is no pair node included. If there is no src-dst pair from the existing migration list included in the MISs of node $v$, we set the second part of $MIGC$ as  
$\max \left( {{1 \mathord{\left/
			{\vphantom {1 {MIGInd}}} \right.
			\kern-\nulldelimiterspace} {MIGInd}}} \right) + 1 $ with current minimum $MIGInd$ value.
Thus, the smaller $MIGC$ of a potential migration, the fewer migration dependencies for the selected migration lists and future selections. Note that we do not need to check $MIGC$ of two migrations with the same node, as the result will be the same.

\subsection{Concurrency-Aware Migration Selector}

In this section, we explain the details of the proposed concurrency-aware migration selector (CAMIG) in Algorithm \ref{alg: camig}. It minimizes resource dependency and migration overheads while achieving the objective of resource management. Given the input of the objective of the dynamic resource management, the objective function, available VMs, candidates source and destination hosts, the networking information monitored by the SDN controller, and the VM and host information, CAMIG will generate the live migration list which consists of the selected VMs and the corresponding destinations.

\begin{algorithm}[t!]
	\caption{CAMIG}\label{alg: camig}
	\KwIn{Performance Objective $Score^*$, protential VMs ${i}$, source $H_s$, dst $H_d$}
	\KwResult{Selected Migration List $M_{mig}$}
	
	\SetKwFunction{cliqG}{\textsc{all}Cliques}
	\SetKwFunction{indG}{\textsc{all}IndepSet}
	\SetKwFunction{depG}{\textsc{Create}depGraph}
	\SetKwFunction{bestVm}{\textsc{Get}MigCandidates}
	\SetKwFunction{MigInter}{\textsc{Update}MigInterference}
	\SetKwFunction{depGUpdate}{\textsc{Update}depGraph}
	Step 1. get node clique matrix\\
	$G_{dep}$, $\{M(v_{sd})\} \gets $ \depG(${H_s},{H_d},k$);\\
	$\{C\} \gets$ \cliqG(${G_{dep}}$);\\
	$x \gets 0$;~
	$M^x \gets \emptyset$;~
	$M^x_{mig} \gets \emptyset$;\\
	\Do{$|\hat {Score}^{x+1} - Scor{e^*}| > \delta$ and $\hat {Score}^{x+1} > \hat {Score}^{x}$ and $x + 1 < |\{m\}|$}{
		Step 2. get candidate VMs\\
		\MigInter(VM$_i$, $H_d^i$, $L_{sd}^i$);\\ 
		$\hat {Score}^{x+1}$, $\{v_{sd}^{j}\}$, $\{m^j_{sd}\}$ $\gets$ \bestVm($p_{current}$, $\{w_i\}$, $\{H_d^i\}$, $Score^x$, $M^x_{mig}$);\\
		
		Step 3. select the optimal migration\\
		$\hat {v}_{sd}^j \gets v_{sd}^0$;~ $\hat m^j \gets m^0_{sd}$;\\
		\If{$|\{v_{sd}^{j}\}| > 1$}{
			\ForEach{$v \in \{v_{sd}\}$}{
				${C^v}$ = \cliqG($\{C\}$, $v$);\\
				${I^v}$ = \indG($G_{dep}$, $\{C\}$, $v$);\\
				\If{$Inte{r^{j,v}} < Inte{r_{\min }}$}{
					$Inte{r_{\min }} \leftarrow Inte{r^{j,v}}$;\\
					$\hat {v}_{sd}^j \gets v_{sd}^j$;~ $\hat m^j \gets m^j_{sd}$;\\	
				}
			}
		}
		${M^{x + 1}} \leftarrow {M^x} \cup \hat v_{sd}^j$; ~ $M^{x+1}_{mig} \gets M^{x}_{mig} \cup \hat m^j_{sd}$;\\
		\depGUpdate($G_{dep}$, $\{C\}$, $\hat m^j_{sd}$, $\hat v_{sd}^j$)\\
	}
	
	\KwRet{$M_{mig}$}
\end{algorithm}

In \textbf{Step 1}, $G_{dep}$ and $M\left( {{v_{sd}}} \right)$ are generated according to Algorithm \ref{alg: depgraph}.
In \textbf{line 3}, we find all maximal cliques $\{C\}$ of $G_{dep}$.
From \textbf{line 5-18}, at each round $x$, we select the optimal migration from src-dst node $\hat v_{sd}^j$ based on both $MIGC$ and single migration overhead $Inte{r_{single}}$.
As a result, it gets the overall minimal dependencies and single overheads of the total migrations to satisfy the objective of the dynamic resource management.
For \textbf{Step 2}, in each optimal round, it first updates the single migration interference of each candidate VM for its potential destinations.
According to the selected migrations of previous rounds $M_{mig}^x$ and current placement, it gets the newest VM to Host mapping.
Then, it obtains the candidate migrations $\{m_{sd}^j\}$ and corresponding pairs ${v_{sd}^j}$ in this round with the same objective score $\hat{score}^{x+1}$.
It can generate more potential migrations by enlarging the score tolerance of the optimal objective in each round.
For \textbf{Step 3}, the optimal migration with the minimum total migration interference $Inte{r_{\min }}$ is selected. 
It first calculates $\{C^v\}$ based on all maximal cliques $\{C\}$ generated based on Bron-Kerbosch Degeneracy algorithm and $\{I^v\}$ according to Algorithm \ref{alg: clique-indep-alg}.
Then, based on the pair list of already selected migrations $M^x$, the migration overhead of migration $m_i$ with src-dst pair $v$ can be calculated as:
\begin{equation}
	Inte{r^{i,v}} = \kappa_{mig} \cdot Inter_{single}^{i,v} + \kappa_{mig} \cdot Inter_{single}^{i,v} \cdot MIGC_v^{{M^x}}
\end{equation}
where $\kappa_{mig}$ is the coefficient for the value normalization of single migration overheads. Then, the single migration overhead $Inter_{single}^{i,v}$ and $MIGC_v^{{M^x}}$ can be calculated based on Equation (\ref{eq:total-com-mem})-(\ref{eq:round-number}) and (\ref{eq: migc1})-(\ref{eq: migind}), respectively.
In \textbf{line 17}, it adds the optimal migration of this round $\hat m^j_{sd}$ and its pair node $\hat v_{sd}^j$ to the currently selected migration list $M^x_{mig}$ and corresponding node list $M^x$.

In \textbf{line 19}, algorithm \textbf{UpdatedepGraph}  updates the dependency graph and all maximal cliques according to the selected migration.
Certain potential migrations related to the selected optimal migration are deleted from the the pair list. For example, in Section \ref{sect: example}, if we choose migration $v_1^1: H1 \to H3$, then $v_1^2: H3 \to H1$ is excluded for future selection. 
Note that we do not need to use Bron-Kerbosch Degeneracy to recalculate $\{C\}$ based on the new subgraph (Theorem \ref{thm: 2}). 
If the pair list is empty after update $M_{sd} = \emptyset$, the corresponding node $v_{sd}$ will be removed from $G_{dep}$ and $\{C\}$. If the updated clique size is 1 and the only one vertex left has connected edge, remove such clique. Duplicated cliques are also removed. 

The stop conditions of CAMIG are: (1) at the round $x$, the currently selected VM migrations achieve the objective of dynamic resource management; (2) the objective is not improved in the last round; (3) round number equals to the total number of potential VMs. 

\newtheorem{thm}{Theorem}
\begin{thm}[Correctness of UpdatedepGraph] \label{thm: 2}
	Given a graph $G=(V,E)$, $V \ne \emptyset$, its all maximal cliques $\{C\}$ and its subgraph $G^{'} = G[V\backslash \{v^{'}\}]$ with removing vertices $\{v^{'}\}$, results of UpdatedepGraph algorithm $\{C^{''}\}$ and listing all maximal cliques $\{C^{'}\}$ of $G^{'}$ are the same.
\end{thm}

\begin{proof}
	Bron-Kerbosch Degeneracy generates all and only maximal cliques $\{C\}$ of $G$ \cite{eppstein2010listing}. 	
	(1) For $\forall {C^{'}}$, $\forall {C^{''}}$, $ |{C^{'}}| = 1$ and $|{C^{''}}| = 1$. 
	Because the $V(G)\backslash \{v^{'}\} = V(G^{'})$.
	Thus, $\{C^{'}\} = \{C^{''}\}$.
	(2) For $\forall {C^{'}}$, $\forall {C^{''}}$, $|{C^{'}}| > 1$ and $|{C^{''}}| > 1$. For the sake of prove, we assume that 
	$\exists{C^{'}}, {C^{'}} \notin \{C^{''}\}$. Then, $\exists {C_e}, \exists {C_e^{'}}$, $ {C_e} = C_e^{'} \cup \{ {v_e}\}  \cup \{ v_e^{'}\} $, 
	where ${C_e} \in \{ C\} $, $C_e^{'} \in \{ {C^{'}}\} $,  part of remaining vertices $\{ {v_e}\}  \subseteq V(G)\backslash \{ {v^{'}}\} $, part of removing vertices $\{ v_e^{'}\}  \subseteq \{ {v^{'}}\} $. 
	Then, we have $C_e^{'} \cup \{ {v_e}\} \in \{C^{''}\}$.
	If $\{ {v_e}\}  \ne \emptyset $, 
	because $\forall C^{'}, \exists C, C^{'} \subseteq C$, then $C_e^{'} \cup \{ {v_e}\} \in \{C^{'}\}$. We have a contradiction, as $C_e^{'}$ is a maximal clique of $G^{'}$. If $\{ {v_e}\}  = \emptyset$ or $C_e = C_e^{'}$, as the UpdatedepGraph removes all $v^{'} \in \{v^{'}\}$, we have a contradiction $C_e^{'} \in \{C^{''}\}$.
	Thus, $\forall C^{'} \in \{C^{''}\}$. Similarly, we can prove $\forall C^{''} \in \{C^{'}\}$. Therefore, $\{C^{'}\} = \{C^{''}\}$.
\end{proof}

The worst-case running time of Bron-Kerbosch Degeneracy is $O\left( {dn{3^{d/3}}} \right)$ \cite{eppstein2010listing} with total $n$ vertices and degeneracy $d$. The upper bound of all maximal cliques/independent sets of a Graph $G$ is $\left( {n - d} \right){3^{d/3}}$. Thus, given $c$ maximal cliques, the time complexity of the algorithm for calculating MIGC is $O(cn)$. Then, the worst-case running time of CAMIG is $O\left( {\left( {n - d} \right){n^2}{3^{d/3}}} \right)$.
We perform extensive computational evaluation on time complexity in Section \ref{sect: time-analysis}. It demonstrates that algorithm CAMIG is very fast in practice.

\section{Performance Evaluation} \label{section: evaluation}
In this section, we evaluate the performance of our proposed concurrency-aware migration selection (CAMIG) algorithm for dynamic resource management with several parameters, such as total migration time, total migration number, and corresponding dynamic resource management performance in load balancing and energy-saving scenarios. We used both real-world workload trace from PlanetLab \cite{park2006comon} and synthetic workloads for the evaluation. We also performed extensive computational experiments for time analysis. The results show that the proposed algorithm can significantly improve the multiple migration performance~\cite{he2021sla} while achieving the target of resource management.

The scalability of Mininet is limited due to the limitation of its resource usage and the operating systems, which prevents the cloud-scale simulations. Furthermore, it can not simulate the computing resource for the dynamic resource management and multiple migration scheduling.
Thus, we have implemented components for the multiple migration scheduling simulations~\cite{he2021sla} based on the CloudSimSDN~\cite{son2019cloudsimsdn}. The accuracy of network processing of CloudSimSDN compared to Mininet is validated in \cite{son2015cloudsimsdn}. Based on the phases of pre-copy migration, the event-driven simulator\footnote{CloudSimMig. \url{https://github.com/hetianzhang/CloudSimMig}} can evaluate the performance of multiple migrations in terms of the total migration time, migration execution time, total transferred data, and downtime. 

\subsection{Load Balancing Scenario}
In this section, we evaluate the impact of migration concurrency during the dynamic resource management on the performance of multiple migration scheduling in load balancing scenarios. The target of the resource management policy in this experiment is to keep the total CPU utilization of each physical host to 50\%. For other solutions besides the optimal, we set the target range of the total CPU utilization from 45\% to 55\%. We compare our algorithm CAMIG with the result of the optimal and other load-balancing algorithms: Sandpiper \cite{wood2009sandpiper}, FFD \cite{verma2008pmapper}, and  iAware \cite{xu2014}. 
We first evaluate algorithms on small-scale experiments with 8 physical hosts in a Fat Tree. Then, we extend the experimental scale for complex scenarios with more resource dependencies. In extensive experiments, by integrating the proposed concurrency-aware algorithm with existing dynamic resource management algorithms, we directly evaluate and illustrate the scheduling performance improvement in multiple migration planning and scheduling algorithm~\cite{he2021sla}.

\begin{figure}[t]
	\centering
	\includegraphics[width=0.8\linewidth]{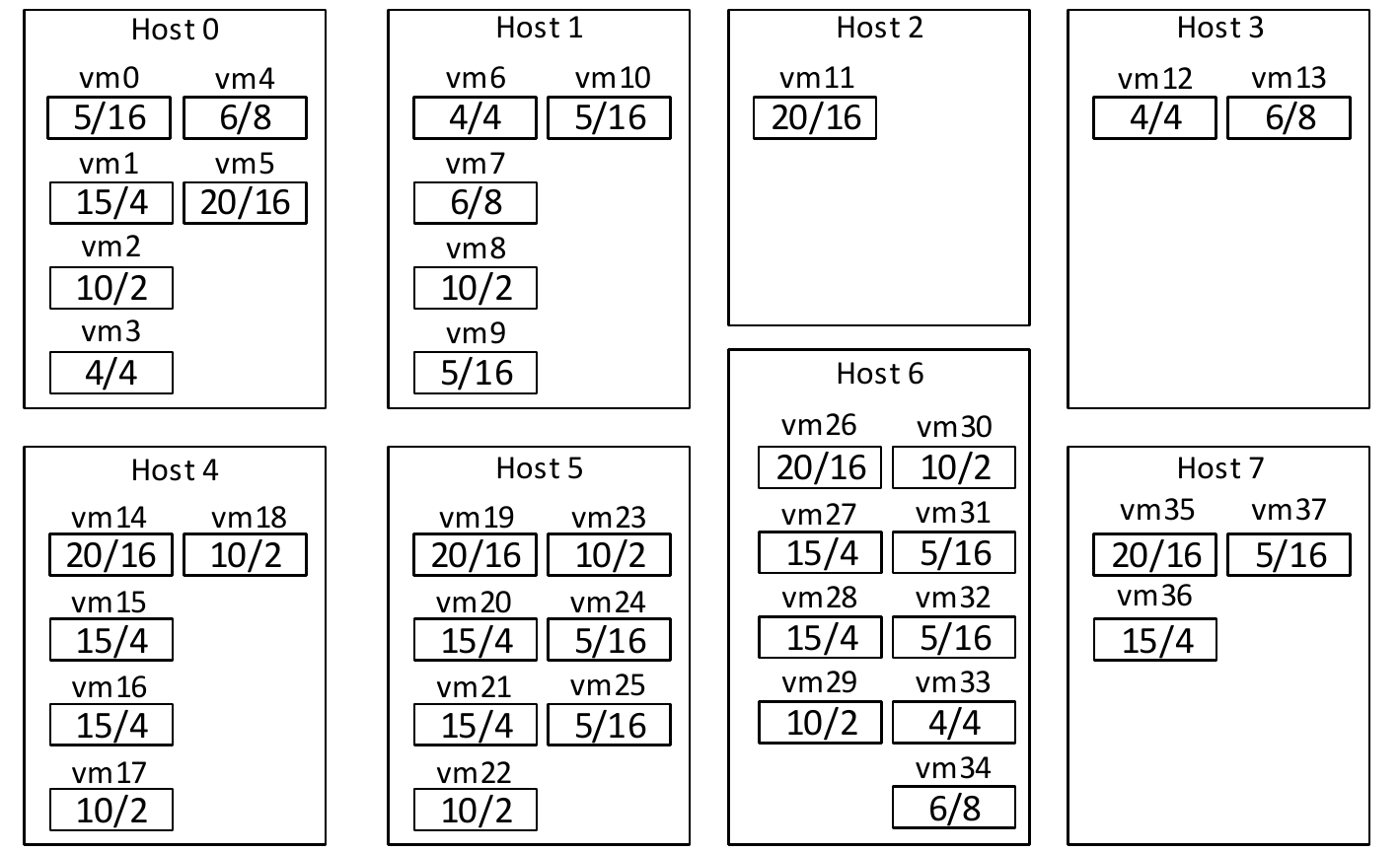}
	\caption{Initial mapping for 8 different physical hosts with CPU utilization(\%)/Requested Memory(GB) }
	\label{fig: small-scall-mapping}
\end{figure}

\subsubsection{Experimental Setup}
In order to focus on the performance of multiple migrations for different migration requests generated by various resource management algorithms, we controlled variables of single migration overheads, such as dirty page rate, that other comparison algorithms ignore. In the load-balancing scenario, we use the same source selection as Sandpiper to choose over-utilized source hosts for potential migration.

The actual location of physical hosts in Fat Tree topology with different resource utilization is generated randomly, which causes different source and destination selections and resource dependencies in each random setup. Without specific explanation, the result is the average value of 10 experiments.
Causing utilization difference among hosts, the initial placement of VMs in each machine with different CPU utilization and memory size is shown in Fig. \ref{fig: small-scall-mapping}. To differentiate the migration value in management objective and migration schedule, we create VMs with different combinations of high, medium, and low value of resource utilization and memory size.
The CPU utilization of each VM is from 4\% to 20\% of the total host CPU resource. As a result, the CPU utilization of each host is from 10\% to 90\%. The Memory size of each VM is from 2 GB to 16 GB, which can result in various migration overheads. 

Other parameters of pre-copy migration are set as the same for each VM.
The dirty page rate factor is 0.001 per second. For example, with a 0.001 per second dirty page rate factor, the dirty page rate of a VM with 16 GB memory is 128 Mbps. The data compression ratio is 0.8. The iteration and downtime threshold is 30 and 0.5 seconds, respectively.We create a k-8 FatTree Data Center Network (128 hosts) with 1 Gbps bandwidth between switches. For the purpose of irrelevant parameter exclusion in experiments, each physical host has 16 CPUs with 10000 MIPS each, 10GB RAM, 1 TB storage, and 1 Gbps network interface. Note that hosts are not required to be identical in the proposed algorithm. 

Dual simplex (Gurobi optimizer 9.0\footnote{Gurobi solver, \url{https://www.gurobi.com/}} and Python-MIP 1.6.7\footnote{Python-MIP. \url{https://github.com/coin-or/python-mip}}) were used to get the optimal solution of the MIP model. We also proposed a baseline algorithm called HostHits (hht).  As shown in CAMIG selections, several potential destinations can achieve the same objective of dynamic resource management. It chooses the least selected/hit host as the destination of VM migration in each migration selection iteration.

For original Sandpiper, FFD and iAware without multiple migration scheduling, the sum of migration execution time is the actual total migration time of these algorithms, because they only consider one-by-one migration scheduling. However, given the multiple migration requests, we apply the multiple migration planning and scheduling algorithm~\cite{he2021sla} to all resource management algorithms in experiments and evaluate and show the results of corresponding performance in multiple migration scheduling.

\begin{table*}[ht] 
	
	\centering
	\caption{Total migration time/sum of migration execution time comparison in the extending mapping scenarios}
	\resizebox{0.8\linewidth}{!}{
		\begin{tabular}{|l|l|l|l|l|}
			\hline
			approach		&	multi1				&	multi2				&	multi3					&	multi4  				\\
			\hline
			optimal 		&	71.5313 / 172.9520	&	71.5313 / 345.9040	&	71.5313 / 518.8560		&	71.5313 / 691.8080 	    \\
			camig			&	86.5060 / 189.5725	&	86.5060 / 379.1451	&	86.5060 / 568.7177		&	86.5060 / 758.2903		\\
			sandpiper		&	86.5060 / 189.5725	&	86.5060 / 379.1451	&	99.4928 / 594.7547		&	99.4860 / 784.4188 		\\
			optimal+sandpiper &	86.5329 / 189.6183	&	86.5329 / 379.2367	&	86.5094 / 568.8412		&	86.5329 / 758.4734	    \\
			ffd				&	73.2070 / 133.0450	&	88.1817 / 266.1101	&	73.2203 / 399.2128		&	88.1949 / 532.3334		\\
			iaware			&	86.5158 / 174.6271	&	578.5142 / 969.6401	&	374.0354 / 1448.9137	&	419.1750 / 1941.2873		\\
			\hline
	\end{tabular}}
	\label{tb: comparison-time}
\end{table*}

\begin{table*}[ht] 	
	\centering
	\caption{Comparison of dependent migrations/multiple migration interference/standard deviation of CPU utilization}
	\resizebox{0.8\linewidth}{!}{
		\begin{tabular}{|l|l|l|l|l|}
			\hline
			approach	&	multi1						&	multi2						&	multi3						&	multi4  				\\
			\hline
			optimal 	&	5/ 3.1648/ 0				&	10/8.9682/ 0				&	15/ 10.2091/ 0				&	20/ 14.3697/ 0 	    \\
			camig		&	10/ 6.2048/ 7.4286			&	20/13.0928/ 6.9333			&	30/ 31.2534/ 6.7826			&	40/ 36.4625/ 6.7097			\\
			sandpiper	&	10/ 6.2048/ 7.1428			&	34/ 22.9404/ 6.6667			&	55/ 58.0650/ 6.6087			&	76/ 70.0414/ 6.5161 		\\
			optimal+sandpiper 	&	10/ 6.8879/ 14.2857 		&	20/ 13.9321/ 10				&	30/ 21.4943/ 8.7826				&	40/ 32.6992/ 9.7419 	    \\
			ffd			&	11/ 6.3697/ 84.5714		&	21/ 19.2937/ 78.9333		&	33/ 23.1770/ 77.2173		&	54/ 45.3416/ 76.3870		\\
			iaware		&	15/ 9.0528/ 35.7142		&	53/ 49.4754/ 210.8			&	48/ 38.6271/ 235.9130			&	79/ 68.3587/ 248.25801	\\
			\hline
	\end{tabular}}
	
	\label{tb: comparison-interference}
\end{table*}

The rationale is that Sandpiper chooses the largest volume/memory VM from one of the most overloaded physical host to minimize live migration overheads. The volume as the multi-dimensional loads indicator is defined as: $Volume = \frac{1}{{(1 - cpu)(1 - net)(1 - mem)}}$ \cite{wood2009sandpiper}, where \textit{cpu}, \textit{net}, and \textit{memory} are normalized utilizations of corresponding resources.
FFD (First-Fit Decreasing) algorithm selects the smallest size VMs from over-utilized hosts and assigns them in the FFD ordering of the spare resources to under-utilized hosts. 
iAware considers both co-location VM interference and the single live migration overheads. The co-location VM interference is linear to the number of VMs one physical machine hosts in Xen. The migration selection in iAware is sequentially decided in each round of the greedy algorithm.

\subsubsection{Scalability Evaluation}
We extend the scale of experiments (multi2, multi3, and multi4) by multiplying the same mapping 2, 3, and 4 times. 
Total of $N$ hosts are randomly placed among the first $N$ number locations in the Fat Tree topology with 128 hosts. Each scenario has 16, 24, 32 candidate destination hosts with a total 76, 114, and 152 potential migration VMs, respectively. For example, the physical Host 16, Host 8 and Host 0 have the same VM initial allocation. However, for each scenario, the placement of each physical host in the FatTree is generated randomly. As the resource management algorithms do not have the prior knowledge of the initial placement, the combination of source, destination, and instances during migration selection is increased exponentially. As a result, with the experiment scale increasing, more random source and destination combinations of potential migrations are generated for each experiment. We conducted 10 experiments in each scenario and show the average results. 
 
Table \ref{tb: comparison-time} and \ref{tb: comparison-interference} show the results of the optimal solution, CAMIG and the optimal solution with Sandpiper VM selection, Sandpiper, FFD, and iAware in total migration time with multiple migration schedule, total migration execution time (one-by-one schedule), the number of dependent migration tasks, multiple migration interference value, and the load-balancing performance (standard deviation of CPU utilization). The multiple migration interference value is the sum of normalized single overheads from dependent migrations. Although all physical hosts are arranged randomly, the optimal result should be the same as in scenario multi1.

\textbf{Analysis:}
Table \ref{tb: comparison-time} and \ref{tb: comparison-interference} show that the MIP model achieves the optimal in all scenarios. With the source host selection from Sandpiper, comparing CAMIG with the optimal solution, as the problem scale increases, CAMIG can maintain the optimal performance in multiple migration scheduling as well as the number of resource-dependent migrations. In multi3 and multi4, CAMIG  over-satisfies the requirement of load-balancing by losing the value of multiple migration interference. For the Sandpiper and iAware, as the the scale of the problem increases, the number of dependent migrations and the value of multiple migration interference increase dramatically, which leads to a larger total migration time in both multiple and one-by-one scheduling. FFD can not satisfy the requirement of load-balancing in the system.

The total migration time of Sandpiper is increased by 15.01\% in multi3 and multi4. In Table \ref{tb: comparison-interference}, although FFD has the lowest total migration time and migration execution time, it cannot achieve the ideal load-balancing performance. The standard deviation of FFD is the largest among other algorithms. Moreover, the largest total migration is increased by 21.33\% compared to the lowest. For iAware, the actual total migration time equals to the total migration execution time by only allowing one-by-one scheduling. With multiple migration scheduling, iAware has the worst performance in total migration time and load-balancing due to the trade-off between migration execution time and co-location interference. The total migration time varies largely in different scenarios, increasing at most 568.68\%.

\begin{figure*}[th]
	\centering
	\begin{subfigure}{.33\linewidth}
		\centering
		\includegraphics[width=\linewidth]{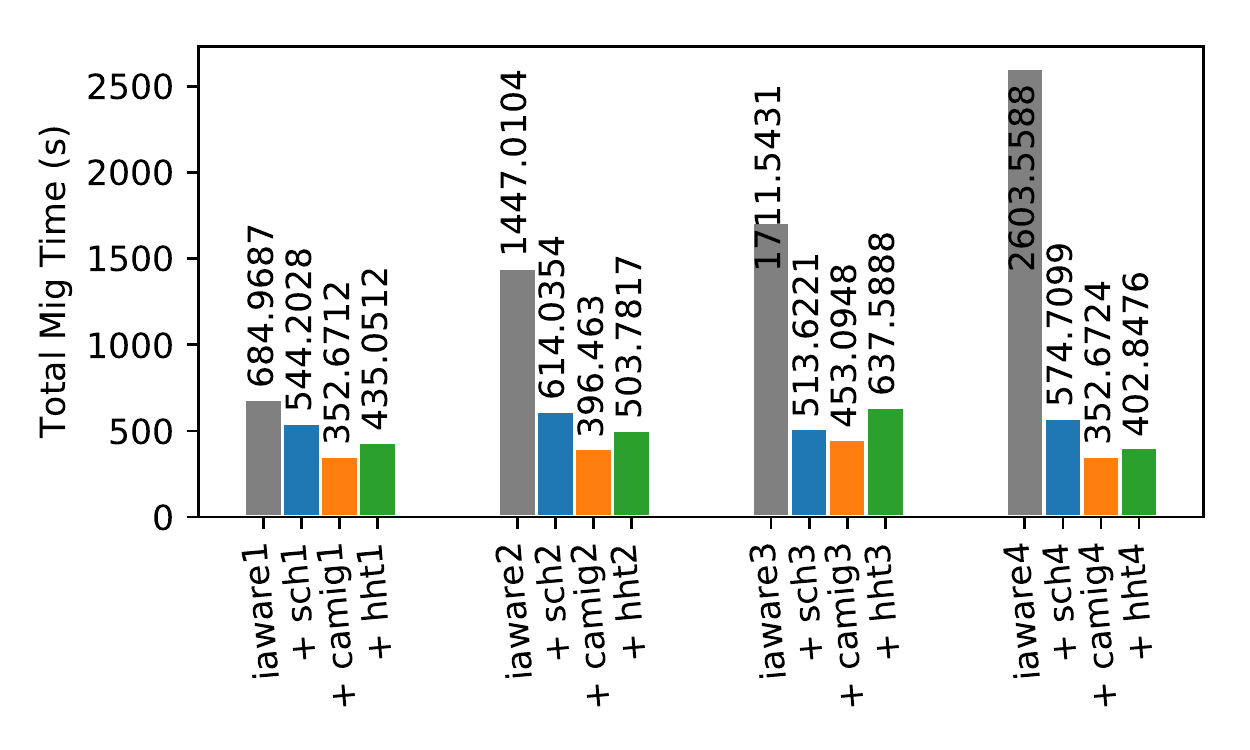}
		\caption{iAware}
		\label{fig: iaware}
	\end{subfigure}%
	\begin{subfigure}{.33\linewidth}
		\centering
		\includegraphics[width=\linewidth]{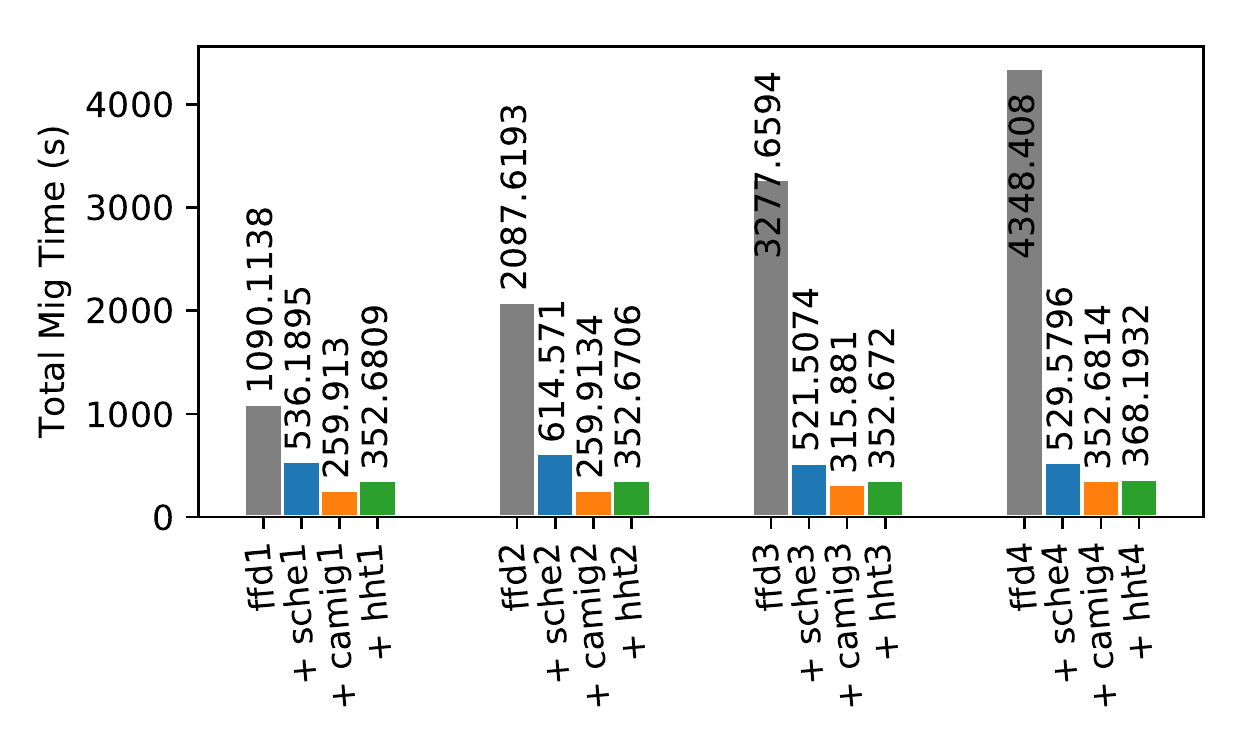}
		\caption{FFD}
		\label{fig: ffd}
	\end{subfigure}
	\begin{subfigure}{.33\linewidth}
		\centering
		\includegraphics[width=\linewidth]{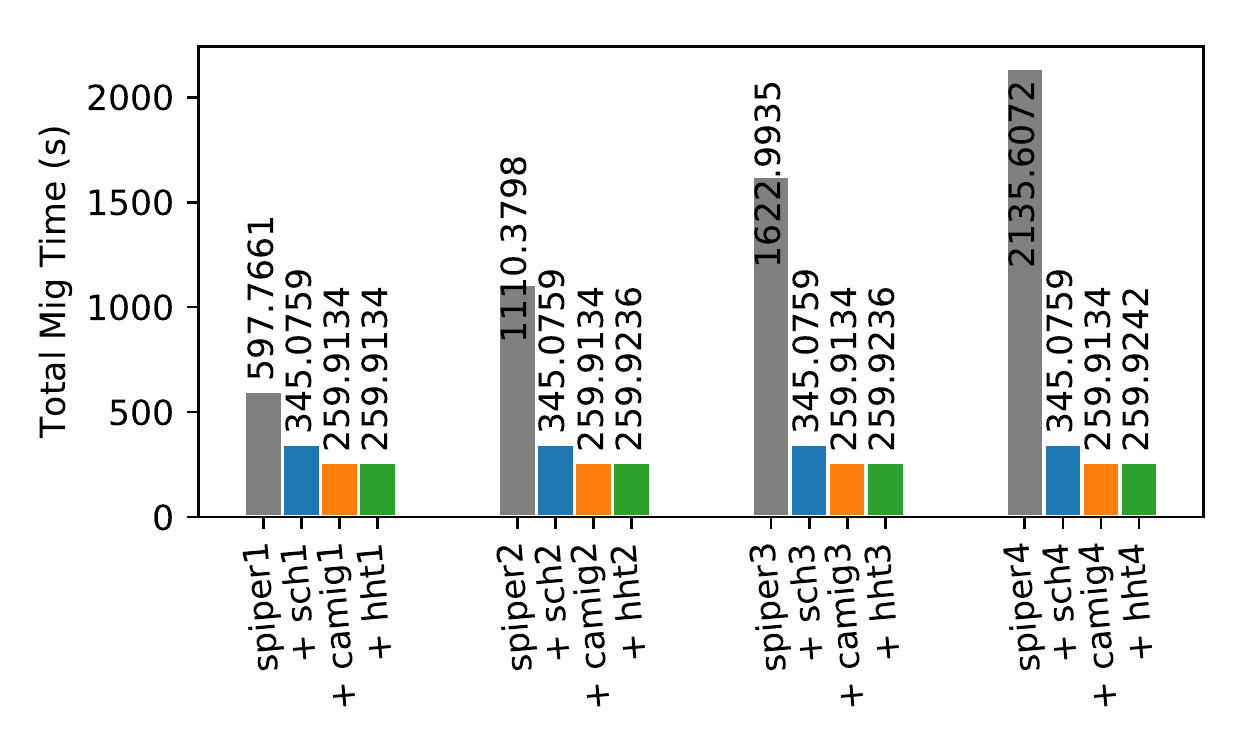}
		\caption{Sandpiper}
		\label{fig: sandpiper}
	\end{subfigure}
	\caption{Performance Comparison with one-by-one, direct multiple scheduling, CAMIG and HostHits}
	\label{fig: camig-compare}
\end{figure*}

\subsubsection{Extensive Evaluation}
As every load-balancing policy has its own logic for VM selection, it is difficult to evaluate the improvement of multiple migration directly. Thus, in this section, we extended the experiments by integrating the HostHits and CAMIG algorithm with the existing policies: iAware, FFD, and Sandpiper. With the benefit of flexibility, CAMIG can be adapted to other existing dynamic resource management algorithms.
We randomly generated VM Memory Size from 8 to 14 GB with the same scenarios (Fig. \ref{fig: small-scall-mapping}). Each result is the average value of 10 experiments in each scenario. Fig. \ref{fig: camig-compare} illustrates the multiple migration performance in total migration time based on the migration requests of these policies with one-by-one scheduling and multiple migration scheduling (+sch), and multiple migration scheduling performance based on the migration requests of CAMIG (+camig) and HostHits (+hht) in 4 different scenarios.

\textbf{Analysis:}
Fig. \ref{fig: iaware} indicates that iAware with CAMIG can achieve the best performance with multiple migration scheduler in all 4 scenarios. The performance is increased by 20.55\%, 57.57\%, 70.02\%, and 77.93\% when migration requests scheduled by the multiple migration scheduler, respectively. However, with CAMIG the performance is increased by 48.54\%, 72.63\%, 73.52\%, and 86.48\% compared to the original iAware and increased by 35.29\%, 35.50\%, 11.89\%, and 38.68\% compared to the performance of iAware with only multiple migration scheduler. Moreover, although iAware with HostHits generally has a better performance compared to iAware+scheduler, as shown in scenario multi3, it results in a worse total migration time due to creating a larger clique of the dependency graph.
For FFD, CAMIG can increase the performance up to 91.90\%, 57.82\%, and 26.42\% compared to FFD with one-by-one scheduler, multiple migration scheduler, and HostHits (Fig. \ref{fig: ffd}). Moreover, Fig. \ref{fig: sandpiper} shows that the performance of Sandpiper with CAMIG in total migration time is increased by up to 87.87\% and 24.68\% than Sandpiper with one-by-one scheduler and multiple migration scheduler, respectively.

\subsubsection{Summary}
In summary, CAMIG can efficiently improve the multiple migration performance while achieving the target of load-balancing resource management. The performance of comparing load-balancing policies can be increased by up to 91.90\%, 57.82\%, and 28.89\% as compared to the one-by-one scheduler, the multiple migration scheduler, and HostHits, respectively. CAMIG outperforms the original policy and the HostHits. The round-robin algorithm HostHits cannot guarantee the multiple migration performance though it generally can decrease the total migration time.

\subsection{Processing Time Analysis} \label{sect: time-analysis}
\begin{figure*}[t]
	\centering
	\begin{minipage}{.3\linewidth}
		\centering
		\includegraphics[width=\linewidth]{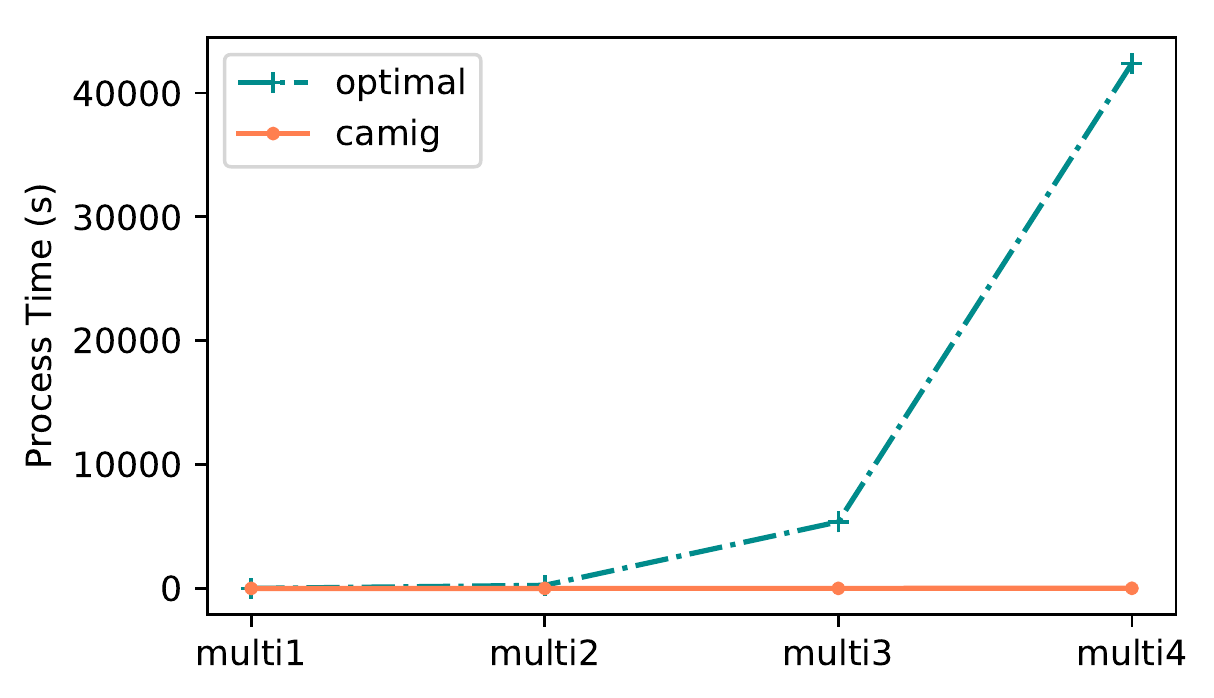}
		\caption{Runtime comparison between optimal and CAMIG}
		\label{fig: runtime-optimal}
	\end{minipage}%
	\hfill
	\begin{minipage}{.3\linewidth}
		\centering
		\includegraphics[width=\linewidth]{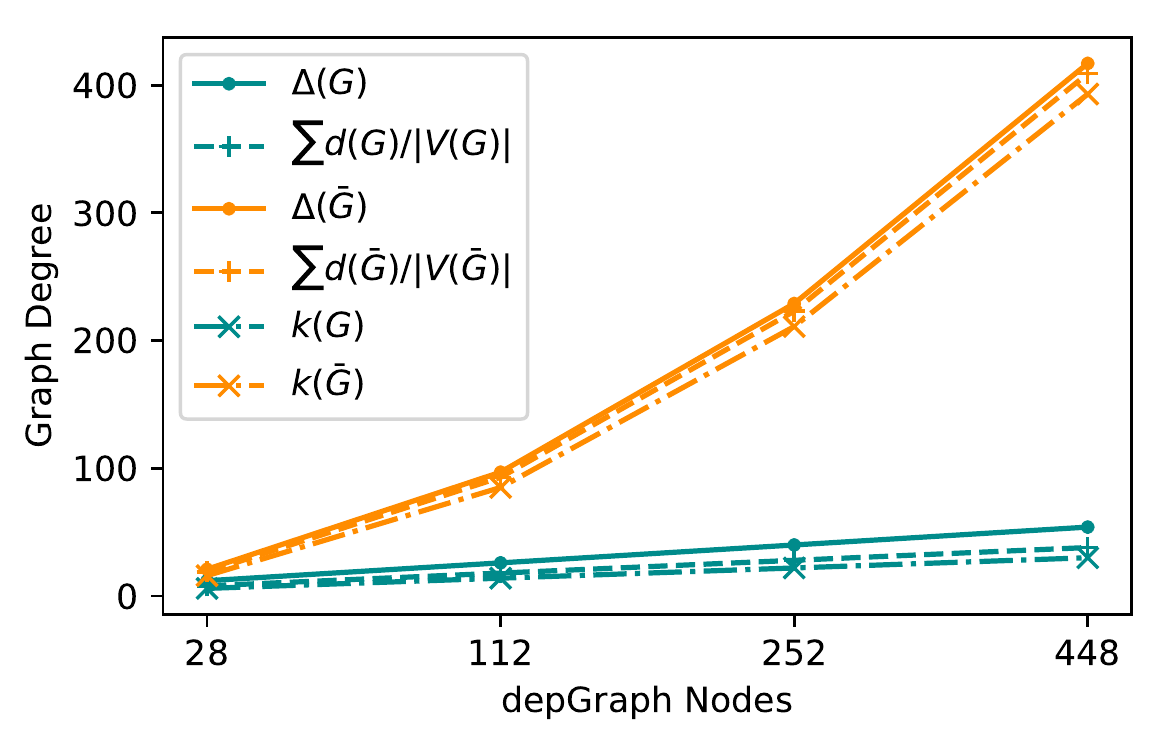}
		\caption{Average and maximum degree and degeneracy}
		\label{fig: depgraph-nodes-degree}
	\end{minipage}
	\hfill
	\begin{minipage}{.3\linewidth}
		\centering
		\includegraphics[width=\linewidth]{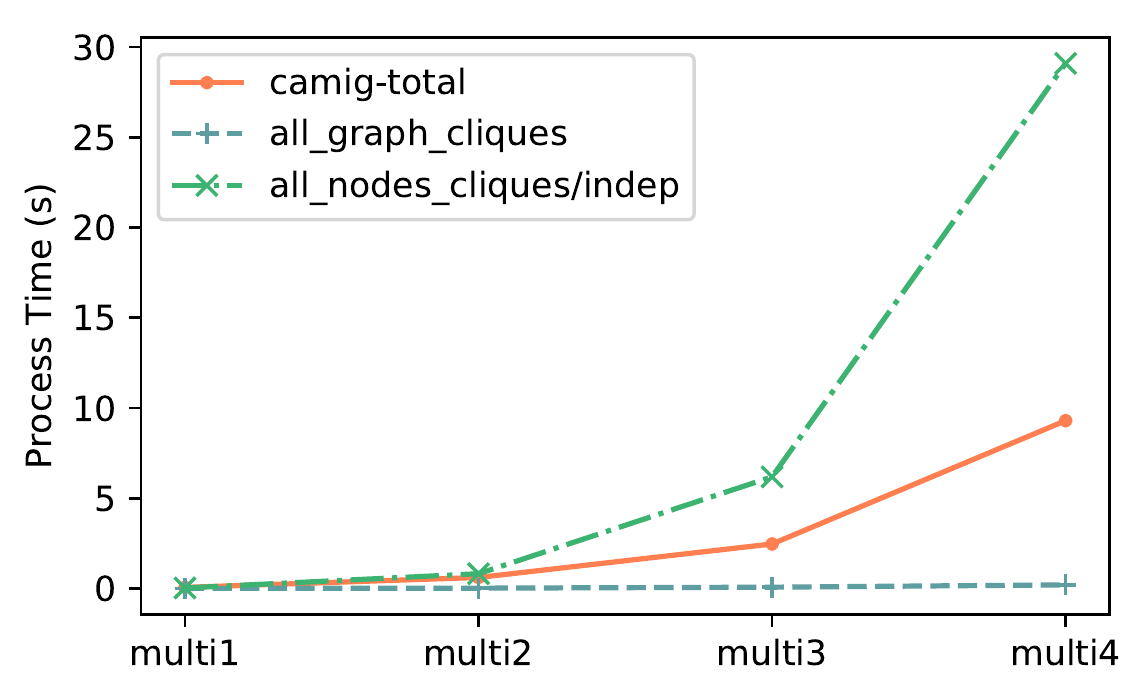}
		\caption{Runtime of CAMIG, all maximal cliques, and all maximal cliques/independent sets of nodes}
		\label{fig: runtime-camig}
	\end{minipage}
	\label{fig: runtime}
\end{figure*}

In this section, we analyze the time complexity of the proposed CAMIG algorithm.
The experiments were run in the computer with i7-7500U CPU with 2.70 GHz, and 15.9 GB RAM in Windows 10 64-bit Operating System. 
Fig. \ref{fig: runtime-optimal} illustrates that the runtime of the optimal solution solved by MIP solver is increased exponentially against the linear growth of the problem size. The runtime of the optimal solution on average is 3.07s, 251.51s, 5373.35s, and 42388.0s in 4 scenarios, respectively. Thus, it is impractical to generate the optimal result when facing the problem in real life. 

Fig. \ref{fig: depgraph-nodes-degree} illustrates the connectivity properties of dependency graph in terms of average degree $\sum {{{d\left( G \right)} \mathord{\left/
{\vphantom {{d\left( G \right)} {\left| {V\left( G \right)} \right|}}} \right.
\kern-\nulldelimiterspace} {\left| {V\left( G \right)} \right|}}} $, maximum degree $\Delta \left( G \right)$, and degeneracy of the dependency $k\left( G \right)$ and its complement $\bar{G}$. The number of maximal cliques is 12, 28, 42, 56 with the degeneracy (a measure of graph spareness) of the dependency graph as 6, 14, 22, 30. Therefore, it is much easier to generate all maximal cliques with a small degeneracy. However, the degeneracy of the complement dependency graph increased dramatically as 16, 85, 211, 393. Thus, it is impractical to generate all maximal cliques of the complement graph as the problem size becomes significantly large. In other words, Bron-Kerbosch Degeneracy algorithm can reach the worst-case runtime when the graph becomes considerably dense. As a result, it can only generate all 661 maximal independent sets in the smallest scale scenario (multi1). Fig. \ref{fig: runtime-camig} shows the runtime comparison of CAMIG in total processing time, finding all maximal cliques, and generating all maximal cliques and independent sets for every node. As shown in Algorithm \ref{alg: clique-indep-alg}, we do not need to calculate all maximal cliques and independent sets of every node in the graph. The all\_nodes\_cliques/indep illustrates the upper-bound of runtime. The processing time of CAMIG is increased linearly against the total src-dst node in resource dependency and the average degree or the degeneracy of the complement of the dependency graph as shown in Fig. \ref{fig: depgraph-nodes-degree}.

\begin{table*}[ht!]
	\centering
	\caption{Performance Comparison between LR-MMT, HostHits, CAMIG in energy-saving scenario}
	\resizebox{\linewidth}{!}{
		\begin{tabular}{|l|lll|cc|ccc|ccc|}
			\hline
			\multirow{2}{*}{algorithm} & \multirow{2}{*}{mig. num} & \multirow{2}{*}{$\sum$ total mig. time} & \multirow{2}{*}{$\sum$ dt. (s)} & \multicolumn{2}{c|}{workload num} & \multicolumn{3}{|c|}{serve time incl. and excl. timeout (s)} & \multicolumn{3}{|c|}{energy cost (Wh)} \\
			\cline{5-12}
			&                           &                                         &                                & total       & timeout       & total excl.    & avg. excl.   & avg. incl.  	& total			 & host             & switch            \\            
			\hline
			NoMig		&	-						&	-							& -			 					&	1506464		&0 		& 11214923.24	&7.44	&-						& 	1733432.22	&	1733432.22		& 0						 \\
			LR-MMT		&	3741					&	28038.66					& 355.079 			 			&	1399857		&106497 & 8700783.51	&6.21	&1105.63				&	470492.05	&	465412.23					 & 5079.82						 \\
			HostHits	&	3680					&	25872.79					& 359.032 			 			&	1416806		&89550	& 9028858.54	&6.37	&447.61			 		&	487254.15	&	481810.21					 & 5443.94						 \\
			CAMIG		&	2534					&	7453.37						& 178.071 			 			&	1458906		&47522	& 9945354.17	&6.82	&80.76			 		&	450966.81	&	447817.74					 & 3149.07						 \\		
			\hline
	\end{tabular}}
	\label{tb: energy-results}
\end{table*}

\begin{figure*}[ht!]
	\centering
	\begin{minipage}{0.48\linewidth}
		\centering
		\includegraphics[width=\linewidth]{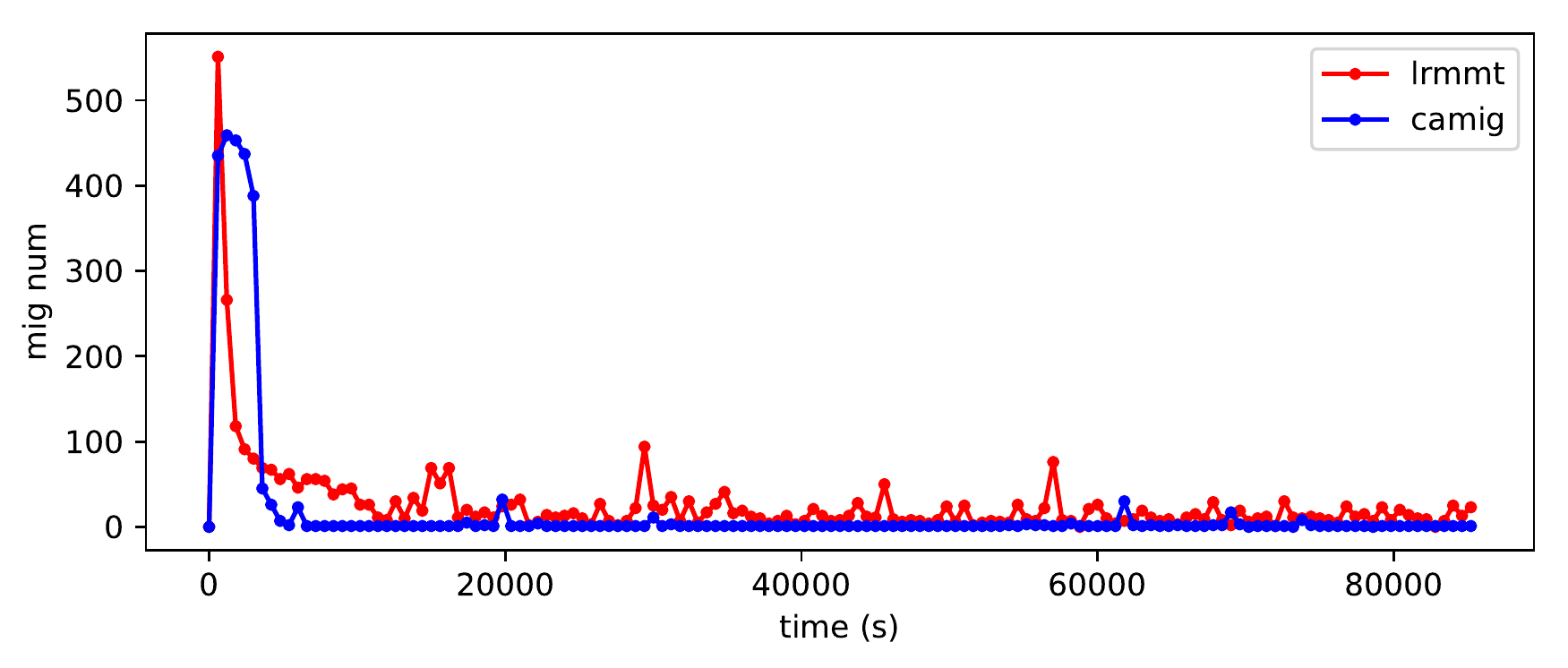}
		\caption{Migration number within each interval}
		\label{fig: energy-mignum}
	\end{minipage}%
	\quad
	\begin{minipage}{0.48\linewidth}
		\centering
		\includegraphics[width=\linewidth]{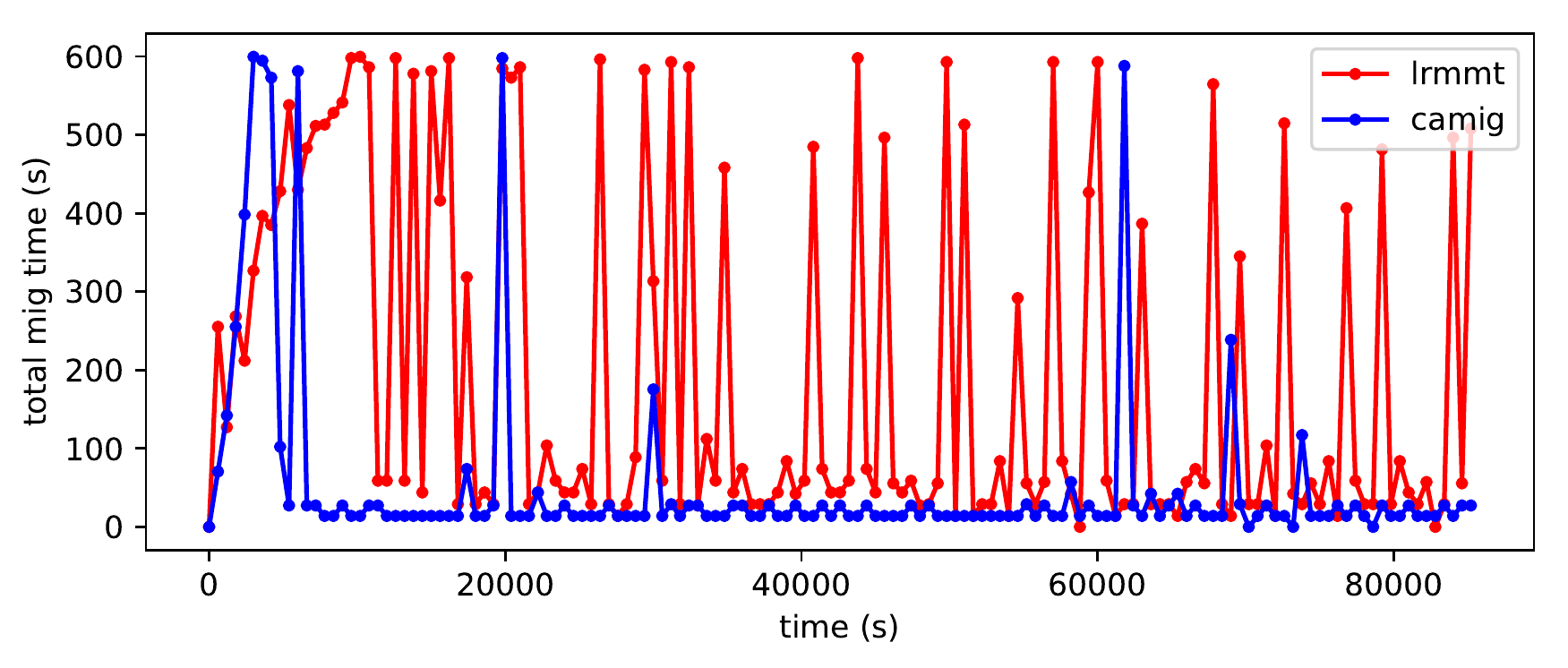}
		\caption{Total migration time within each interval}
		\label{fig: energy-migtime}
	\end{minipage}
	\label{fig: energy}
\end{figure*}
\subsection{Long-term Energy Saving Scenario}
To evaluate the proposed algorithm with the real-world long-term workloads \cite{park2006comon}, we compared CAMIG with LR-MMT \cite{beloglazov2012optimal} in the energy-saving scenario in terms of total migration time, migration numbers, downtime, total/average CPU serve time with and without the timeout workloads, and energy (power) cost of both hosts and switches. 

\subsubsection{Evaluation Configuration}
For the long-term experiments, we created a k-16 FatTree topology (1024 hosts) with 1 Gbps physical links between switches to simulate the environment with limited network resources for live migrations. Each physical host has 8 CPUs with 4000 MIPS, 1024 GB Memory size, 1000 GB Storage, and 1 Gbps network interface.
The real-world workload trace of CPU utilization from Planetlab \cite{park2006comon} was used for the experiments running in 24 hours. There are 1052 CPU utilization files mapping to the same amount of VMs. We generated the workloads based on the MIPS requirement and the CPU utilization varied along the time. In order to illustrate the influence of multiple migration performance, there is no application traffic between different VMs other than the migration flows. There are 4 flavors of VM: 2 vCPUs, [2500, 2000, 1000, 1000] MIPS, [2, 4, 4, 2] GB RAM, 100 Mbps virtual bandwidth, and 4 GB Disk Size. The initial placement of VMs are allocated based on the optimization criteria defined by LR-MMT~\cite{beloglazov2012optimal}.

The LR-MMT algorithm utilizes the Local Regression (LR) method to predict overloading hosts in the upcoming monitor interval. Minimum Migration Time (MMT) policy is used for VM selection to minimize migration overheads. During each monitoring interval of dynamic resource management, CAMIG, as a flexible algorithm, utilizes the same local regression to detect over/under-utilized hosts. In LR-MMT, though there are many equivalent optimal destinations, it only chooses the first fit. For the sake of fair comparison, destination candidates used in CAMIG are provided by the same energy-saving policy in LR-MMT.

\subsubsection{Evaluation Results}
As shown in Table \ref{tb: energy-results}, CAMIG algorithm outperforms both LR-MMT and HostHits. The total energy consumption under no dynamic resource management is 1733432.22 Wh. The LR-MMT algorithm saves 72.86\% energy consumption. Comparing CAMIG with LR-MMT, the host and switch energy consumptions are 3.78\% and 38.01\% less, respectively. The total migration number is 32.26\% less, the sum of total migration time of each monitoring interval is 73.42\% less, the total downtime is 49.85\% less than the LR-MMT algorithm. The performance improvements in total migration time also result in fewer workload timeouts and CPU resource shortages. For VM processing, the average CPU server time is 92.70\% less when there is no timeout mechanism. With a timeout mechanism, CAMIG also reduces the workload timeout by 14.30\% compared to the LR-MMT.

As the sum of total migration time and total migration time of each monitoring interval shown in Table \ref{tb: energy-results} and Fig. \ref{fig: energy-migtime}, within the 24 hours experiment, the performance of CAMIG in multiple migration scheduling is largely better than the LR-MMT. A shorter total migration time during each monitoring interval means a quicker state convergence for minimizing the over-utilization period and maximizing the energy-saving through VM consolidation for under-utilizing hosts. In other words, minimizing the dependencies among multiple migrations is not only critical for the migration scheduling, but also for the dynamic resource management that provides the migration list.  

During the experiments, we find out that there are relatively large equivalent destination candidates in terms of energy saving. Therefore, by exploring the concurrency score among these candidates, we can minimize the resource dependencies among the migrations. As shown in Fig. \ref{fig: energy-mignum}, there are more migrations in CAMIG from 1200s to 3600s than LR-MMT. It is because in LR-MMT once the candidate is used it will be excluded from the remaining destinations. However, by choosing equivalent hosts during the destination selection, CAMIG algorithm enables more available destinations for VMs which need to be migrated from both under and over-utilized hosts. Thus, CAMIG algorithm actually produces fewer migrations in the remaining monitor intervals. It also illustrates that in some cases even the total migration number of CAMIG is larger, the total migration time is much smaller due to the minimum dependency among the migrations. Fig. \ref{fig: energy-migtime} shows that, under certain circumstances (the peak migration time at 20000 second), even if there is a small number of migration tasks, the total migration time is still very large. Due to the nature of the consolidation algorithm, there are many migration tasks sharing the same destination or source hosts. Therefore, in traditional architectures, such as FatTree or even the dedicated migration network, it is inevitable that the convergence of multiple migrations is slower.
As a result, the performance of multiple migration scheduling may be limited by this nature of resource competition among the consolidating VM migrations. 

In summary, the evaluation demonstrates that, CAMIG can efficiently minimize the resource dependency among multiple migration tasks and achieve the objective of dynamic resource management in the long run.  Thus, it also improves the performance of dynamic resource management algorithms in terms of QoS and energy consumption.

\section{Conclusions} \label{section: conclusion}
To the best of our knowledge, we are the first to consider the problem of minimizing the resource dependency of migration requests in dynamic resource management. We formally established a MIP model for the problem and proposed generic concurrency-aware migration selection algorithm (CAMIG). We conducted experiments to compare our proposed algorithms with existing dynamic resource management policies in load balancing and energy-saving scenarios by using both random synthetic setup and real trace data. Without changing the framework of existing policies, the results indicate that CAMIG can largely improve the performance of multiple migrations by up to 91.90\% while achieving the target of dynamic resource management efficiently with near-linear computation growth in practice. In the long-term experiments, it can also reduce the total migration number, service downtime and management target in the host and switch energy consumptions.

\appendices

\section*{Acknowledgments}
This work is partially supported by an Australian Research Council (ARC) Discovery Project (ID: DP160102414) and a China Scholarship Council - University of Melbourne PhD Scholarship. We thank Editor-in-Chief, Associate Editor, anonymous reviewers, Redowan Mahmud, Linnan Ruan, Tawfiqul Islam, and Shashikant Ilager for their valuable comments and suggestions to help improve the paper.

\bibliographystyle{IEEEtran}
\bibliographystyle{abbrv}
\bibliography{ref}

\begin{IEEEbiography}[{\includegraphics[width=1in,height=1.25in,clip,keepaspectratio]{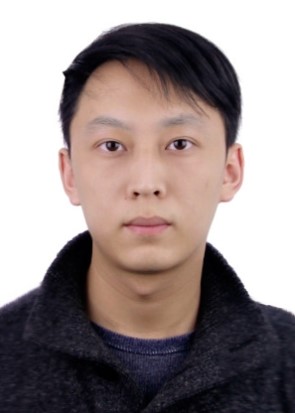}}]{TianZhang He}
	received the B.S. and M.S. degrees from Northeastern University, China, in 2014 and 2017, respectively.
	He is working towards the Ph.D. degree at the Cloud Computing and Distributed Systems (CLOUDS) Laboratory,
	University of Melbourne, Australia. His research interests include Software-Defined Networking, Edge and Cloud Computing, and Network Function Virtualization.
\end{IEEEbiography}
\begin{IEEEbiography}[{\includegraphics[width=1in,height=1.25in,clip,keepaspectratio]{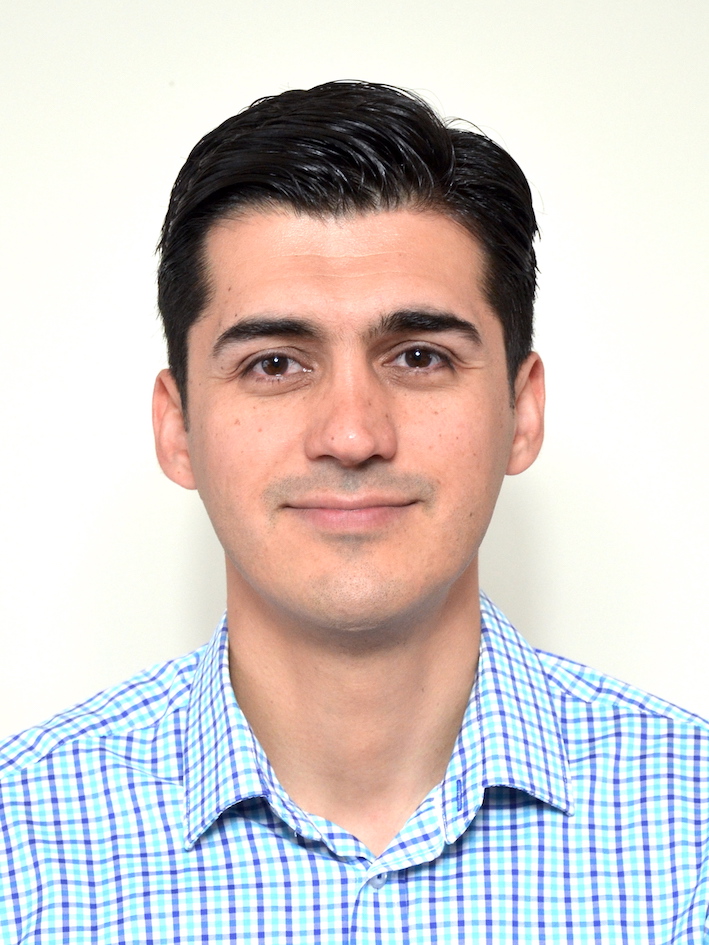}}]{Adel N. Toosi}
	is a lecturer (a.k.a. Assistant Professor) at the Department of Software Systems and Cybersecurity, Faculty of Information Technology, Monash University, Australia. Before joining Monash, Dr Toosi was a Postdoctoral Research Fellow at the University of Melbourne from 2015 to 2018. He received his Ph.D. degree from the School of Computing and Information Systems at the University of Melbourne in 2015.  
	His research interests include Cloud/Fog/Edge Computing, Software-Defined Networking, Green Computing and Energy Efficiency. 
	For further information, please visit his homepage: http://adelnadjarantoosi.info/.
\end{IEEEbiography}
\begin{IEEEbiography}[{\includegraphics[width=1in,height=1.25in,clip,keepaspectratio]{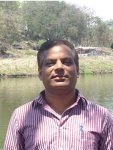}}]{Rajkumar Buyya}
	is a Redmond Barry distinguished professor and the
	director with the Cloud Computing and Distributed Systems (CLOUDS) Laboratory, University of Melbourne, Australia. 
	He has authored over 625 publications and seven text books including ``Mastering Cloud
	Computing" published by McGraw Hill, China Machine Press, and Morgan Kaufmann for Indian, Chinese and international markets, respectively. 
	He is one of the highly cited authors in computer science and software engineering
	worldwide (h-index=149, g-index=322, 116,400+ citations).
\end{IEEEbiography}
\vfill
\vfill
\vfill

\end{document}